\documentclass[11pt,peerreview,draftclsnofoot,letterpaper,oneside,onecolumn]{IEEEtran}
\usepackage{amsmath}
\usepackage{color}
\usepackage{cases}
\usepackage{amsfonts,amsmath,amssymb}
\usepackage{mathrsfs}
\usepackage{cite}
\usepackage{graphicx}
\usepackage{url}
\usepackage{enumerate}

\newtheorem{theorem}{Theorem}
\newtheorem{proposition}{Proposition}

\newtheorem{corollary}{Corollary}
\newtheorem{example}{Example}
\newtheorem{remark}{Remark}

\begin{document}
\baselineskip 3.8ex

\title{Multi-Stage Robust Chinese  Remainder Theorem}


\author{\mbox{Li Xiao$^*$, Xiang-Gen Xia$^*$, and Wenjie Wang$^\dag$}
\thanks{$^*$Department of
Electrical and Computer Engineering,
University of Delaware, Newark, DE 19716 U.S.A. (e-mail:
\{lixiao, xxia\}@ee.udel.edu). Their work was supported in part
by the Air Force Office of Scientific Research (AFOSR) under
Grant FA9550-12-1-0055.}
\thanks{$^\dag$School of Electronics and Information Engineering,
Xi'an Jiaotong University, Xi'an, 710049, China (e-mail:
wjwang@xjtu.edu.cn).
 His work was supported in part by the Natural Science
Foundation of China (NSFC) No. 60971113.}
}

\maketitle

\begin{abstract}
It is well-known that the traditional Chinese remainder theorem (CRT) is
not robust in the sense that a small error in a remainder may cause
a large error in the reconstruction solution.
A robust CRT was recently proposed for a special case when the greatest
common divisor (gcd) of all the moduli
is more than $1$ and the remaining integers factorized by the gcd
of all the moduli are co-prime.
In this special case, a closed-form reconstruction
from erroneous remainders was proposed and a necessary and
sufficient condition on the remainder errors was obtained.
It basically says that the reconstruction error
is upper bounded by the remainder error level $\tau$
if $\tau$ is smaller than a quarter of
the gcd of all the moduli.
In this paper, we consider the robust reconstruction
problem for a  general set of moduli. We first
present a necessary and sufficient condition
for the remainder errors for a robust reconstruction
from erroneous remainders with a general set of muduli
and also a corresponding robust reconstruction method.
This can be thought of as a single stage robust CRT.
We then propose a two-stage robust CRT by grouping the moduli
into several groups as follows. First, the single stage robust
CRT is applied to each group. Then, with these robust reconstructions
from all the groups, the single stage robust  CRT is applied again
across the groups.
This is then easily generalized to multi-stage robust CRT.
Interestingly, with this two-stage robust
CRT, the robust reconstruction holds
even when the remainder error level $\tau$ is above the quarter of
the gcd of all the moduli.
In this paper, we also propose an algorithm on how to
group a set of moduli for a better reconstruction
robustness of the two-stage robust CRT  in some special cases.
\end{abstract}


\section{Introduction}
The problem of reconstructing a large integer from its several remainders
modulo several smaller positive integers (called
moduli) may occur in many applications, such as phase
unwrapping in radar signal processing \cite{zhuang,xu,jorgensen,gwang,ruegg,yimin,dias,qi,gangli1,gangli2,xwli3} and frequency determination
from several undersampled waveforms \cite{xia1,kind1c}.
The traditional solution for this problem is the Chinese remainder
theorem (CRT), see for example, \cite{CRT1,CRT2}, that uniquely formulates
the solution from the remainders if all the moduli are co-prime
and the large integer is less than the product of all the moduli.
When the moduli are not co-prime, the large integer can be uniquely
determined if it is less than the least common multiple (lcm)
of all the moduli \cite{ore,newcrt}, where one may also find the
reconstruction methods. However, it is well-known that
the above solution is not robust in the sense
that a small error in a remainder may cause
a large error in the reconstruction solution,
which may degrade the performance of its applications in
phase unwrapping and frequency determination, since in these
applications, signals are usually noisy and the detected remainders
may be erroneous. For the robustness, there have been several studies recently
\cite{xia,xwli2,wjwang}.
Robust reconstructions from erroneous remainders
were recently proposed in \cite{xwli2,wjwang}
 for a special case when the greatest
common divisor (gcd) of all the moduli
is more than $1$ and the remaining integers factorized by the gcd
of all the moduli are co-prime.
In this special case, a closed-form reconstruction
from erroneous remainders was proposed in \cite{wjwang}
 and a necessary and
sufficient condition on the remainder errors
was also obtained in \cite{wjwang}.
It basically says that the reconstruction error
is upper bounded by the remainder error level $\tau$
if $\tau$ is smaller than a quarter of
the gcd of all the moduli \cite{xwli2}.
A special version of this result was obtained earlier in \cite{zhuang}. In some applications, an unknown, such as the phase unwrapping and frequency estimation, is real valued in general. So, in \cite{wjwang} the closed-form robust CRT algorithm was naturally generalized to real numbers. 
Also, a lattice based method was proposed in
\cite{wenchao} to address the problem of 
estimating a real  unknown distance  with a closed-form algorithm 
 using phase measurements taken at multiple co-prime wavelengths.
 One can see that there are  constraints on the moduli in previous works. 
The constraints on the moduli may, however, limit the robustness
when the range (called dynamic range)
of the determinable integers is roughly fixed.

Different from robustly reconstructing the large integer from its erroneous
remainders, another existing approach is to accurately determine
the large integer by using some of the error-free remainders among all the
remainders \cite{kind1a,kind1d,kind1c},
which may require that significantly many remainders are error-free
and a large number of moduli/remainders may be needed.
This approach may sacrifice the dynamic range for a given set of moduli
(or undersampling rates \cite{kind1c}) and furthermore, in some signal
processing applications,
to obtain error-free remainders may not be even possible,
because observed signals are usually noisy.
A  probabilistic
approach to deal with noises in CRT was proposed in \cite{kind1b},
where all the moduli are required to be primes.

In this paper, we consider the robust reconstruction
problem for a  general set of moduli on which  the constraint
used in \cite{xwli2,wjwang} is no longer required.
 We first
present a necessary and sufficient condition
for the remainder errors for a robust reconstruction
from erroneous remainders with a general set of muduli,
where a reconstruction method is also proposed.
This can be thought of as a single stage robust CRT.
We then propose a two-stage robust CRT by grouping the moduli
into several smaller groups as follows. First, the robust single stage
CRT is applied to each group. Then, with these robust reconstructions
from all the groups, the robust single stage CRT is applied again
across the groups. Interestingly, with this two-stage robust
CRT, the robust reconstruction holds
even when the remainder error level $\tau$ is above the quarter of
the gcd of all the moduli.
The two-stage robust CRT
 is then easily generalized to multi-stage robust CRT.
In this paper, we also propose an algorithm on how to
group a set of moduli for the better reconstruction
robustness of the two-stage robust CRT in some special cases.

Note that the two-stage robust CRT is first appeared in
\cite{binyang} that is, however, based on the special single stage
robust CRT in \cite{wjwang} when the remaining factors
of all the moduli after factorizing out their gcd are co-prime.
With the two-stage robust CRT obtained in \cite{binyang}, the remainder error
level $\tau$ is, in fact, not better than the quarter of the
gcd of all the moduli.
In contrast,
our newly proposed two-stage or multi-stage robust CRT in this paper
is based on the generalized single stage robust CRT for arbitrary
moduli also newly obtained in this paper
 and as mentioned earlier, the remainder error level
$\tau$ can be above the quarter of the
gcd of all the moduli, i.e., it achieves a better robustness bound
than \cite{binyang} does.

This paper is organized as follows.
In Section \ref{sec2}, we first briefly
introduce the robust CRT results obtained in \cite{xwli2,wjwang}.
We then propose our new single stage robust CRT
with the necessary and sufficient condition for a general set of moduli.
In Section \ref{sec3}, we propose  two-stage and multi-stage robust CRT.
In Section \ref{sec4}, we propose an algorithm on how to group
a set of moduli for a better reconstruction robustness of the two-stage robust CRT. In Section \ref{sec5}, we present some simulation results on estimating integers with a general set of moduli. In Section \ref{sec6}, we conclude this paper.

\section{Single Stage Robust CRT}\label{sec2}

Let us first see the robust remaindering problem.
Let $N$ be a positive integer, $0<M_1<M_2<\cdots<M_L$ be $L$ moduli, and $r_1, r_2, \cdots, r_L$ be the  $L$ remainders of $N$, i.e.,
\begin{equation}\label{CRT}
N \equiv r_i \mbox{ mod } M_i \:\mbox{ or }\: N=n_iM_i+r_i,
\end{equation}
where $0\leq r_i<M_i$ and $n_i$ is an unknown integer, for $1\leq i\leq L$.
It is not hard to see that
$N$ can be uniquely reconstructed from its $L$ remainders $r_i$ if and only if $0\leq N< \mbox{lcm}\left(M_1, M_2, \cdots, M_L\right)$. If all the moduli $M_i$ are co-prime, then the Chinese Remainder Theorem
(CRT) provides  a simple reconstruction formula \cite{CRT1}, \cite{CRT2}.

The problem we are interested in this paper is how to robustly reconstruct $N$ when the remainders $r_i$ have errors:
\begin{equation}\label{remainderbound}
0 \leq \tilde{r}_i\leq M_i-1 \:\mbox{ and }\: |\tilde{r}_i-r_i| \leq \tau,
\end{equation}
where $\tau$ is an error level that may be
 determined by, for example,
 the signal-to-noise ratio (SNR) and is also called remainder error bound.
Now we want to reconstruct $N$ from these erroneous remainders $\tilde{r}_i$ and the known moduli $M_i$. The basic idea for the robust CRT
in the recent studies and also this paper
is to accurately determine the unknown integers $n_i$ in
 (\ref{CRT}) which are the folding numbers that may cause large errors in the reconstructions if they are erroneous. Therefore, the problem is to
correctly determine the folding numbers $n_i$ from these erroneous remainders $\tilde{r}_i$.

Once $n_i$ for $1\leq i\leq L$ are correctly found, an
 estimate of $N$ can be given by
\begin{equation}\label{estimatevalue}
\hat{N}(i)=n_iM_i+\tilde{r}_i=
n_iM_i+r_i+\Delta r_i, \mbox{ for any } i, \:\:1\leq i\leq L,
\end{equation}
where $\Delta r_i=\tilde{r}_i-r_i$ denote the errors of the remainders. From (\ref{remainderbound}), $|\Delta r_i|\leq \tau$. Then,
an estimate of the unknown parameter $N$ is
the average of $\hat{N}(i)$:
\begin{equation}\label{averagen}
\begin{array}{lll}
\hat{N}&=&\displaystyle\left[\frac{1}{L}\sum_{i=1}^{L}\hat{N}(i)\right]\\
&=&\displaystyle N+\left[\frac{1}{L}\sum_{i=1}^{L}\Delta r_i\right]\\
&=&\displaystyle N+\Delta \bar{r},
\end{array}
\end{equation}
where $\Delta \bar{r}$ is the average of the remainder errors, and $\left[\cdot\right]$ stands for the rounding integer, i.e., for any $x\in \mathbb{R}$
(the set of all reals),  $\left[x \right]$ is an integer and subject to
\begin{equation}\label{round integer}
-\frac{1}{2}\leq x-\left[x \right]<\frac{1}{2}.
\end{equation}
Clearly in this way the error of the above estimate of $N$
is upper bounded by
\begin{equation}\label{bound of fre}
|\hat{N}-N|\leq \tau,
\end{equation}
i.e., $\hat{N}$ is a robust estimate of $N$.

For the above robust remaindering problem,
solutions, i.e.,
robust reconstruction algorithms,  have been proposed in
 \cite{xwli2, wjwang} for a special case
when the gcd of all the moduli is more than $1$ and the remaining integers factorized by the gcd of all the moduli are co-prime. The main results
can be briefly described below.

Let $M$ be the gcd of all the moduli $M_i$ in (\ref{CRT}). Then $M_i=M\Gamma_i$, $1\leq i\leq L$,
 and assume that
all $\Gamma_i$ for $1\leq i\leq L$ are co-prime, i.e., the gcd of any pair $\Gamma_i$ and $\Gamma_j$ for $i\neq j$ is $1$. Define $\Gamma\triangleq \Gamma_1\Gamma_2\cdots\Gamma_L$. For $1\leq i\leq L$, let
\begin{equation}
\gamma_i\triangleq \Gamma_1\cdots\Gamma_{i-1}\Gamma_{i+1}\cdots\Gamma_L=\Gamma/\Gamma_i,
\end{equation}
where $\gamma_1\triangleq \Gamma_2\cdots\Gamma_L$ and $\gamma_L\triangleq \Gamma_1\cdots\Gamma_{L-1}$. We now show how to accurately determine the folding numbers $n_i$  in \cite{xwli2} and \cite{wjwang}, respectively. First, define
\begin{equation}
S_i\triangleq \Big\{(\bar{n}_1, \bar{n}_i)={\arg\min}_{\begin{subarray}{ll}
\hat{n}_1=0, 1, \cdots, \gamma_1-1\\
\hat{n}_i=0, 1, \cdots, \gamma_i-1
\end{subarray}
}|\hat{n}_iM_i+\tilde{r}_i-\hat{n}_1M_1-\tilde{r}_1|\Big\}.
\end{equation}
Let $S_{i, 1}$ denote the set of all the first components $\bar{n}_1$ of the pairs $(\bar{n}_1, \bar{n}_i)$ in set $S_i$, i.e.,
\begin{equation}
S_{i, 1}\triangleq\{\bar{n}_1| (\bar{n}_1, \bar{n}_i)\in S_i \:\:\mbox{for some} \:\:\bar{n}_i\}
\end{equation}
and define
\begin{equation}
S\triangleq \bigcap_{i=2}^{L}S_{i, 1}.
\end{equation}
It is proved in \cite{xwli2} that if the remainder error bound $\tau$ is smaller than a quarter of $M$, i.e., $\tau<M/4$,
 the folding numbers $n_i$ for $1\leq i\leq L$ can be accurately determined from $S$ and $S_i$. Set $S$ defined above contains only one element $n_1$, and furthermore if $(n_1, \bar{n}_i)\in S_i$, then $\bar{n}_i=n_i$.
In addition, \cite{xwli2} has proposed a $1$-D searching method with the order of $2(L-1)\Gamma_i$ searches.
When $L$ or $\Gamma_i$ gets large, the searching complexity is still high.
Then, a closed-form robust CRT algorithm and its necessary and sufficient condition for it to hold have been proposed in \cite{wjwang}.
For the closed-form algorithm, we refer the reader to \cite{wjwang}
with which the following necessary and sufficient condition
for the accurate determination of the folding numbers $n_i$
is obtained in \cite{wjwang}.

\begin{proposition}\label{wjwtheorem}
\cite{wjwang}
Assume that all $\Gamma_i$, for $1\leq i\leq L$ are  co-prime and
\begin{equation}
0\leq N< \mbox{lcm}\left(M_1, M_2, \cdots, M_L\right)=M\Gamma_1\Gamma_2\cdots\Gamma_L.
\end{equation}
Then, with the closed-form algorithm determining $\hat{n}_i$ for $1\leq i \leq L$ in
\cite{wjwang},
$\hat{n}_i=n_i$ for $1\leq i \leq L$,
i.e., the folding numbers $n_i$ for $1\leq i \leq L$ can be accurately
determined,
 if and only if
\begin{equation}\label{conditionwjw}
-M/2\leq \Delta r_i-\Delta r_1< M/2, \:\mbox{for all} \:\:2\leq i \leq L.
\end{equation}
\end{proposition}

Although the condition (\ref{conditionwjw}) in Proposition \ref{wjwtheorem} is necessary and sufficient for the uniqueness of the solution of the folding numbers $n_i$, it involves with two remainder errors and is hard to check in practice. However, with this result the following proposition  becomes obvious,
which coincides with the much simpler sufficient condition
 in \cite{xwli2}.
\begin{proposition}\label{wjwcorollary}
\cite{xwli2, wjwang}
Assume that all $\Gamma_i$ for $1\leq i\leq L$ are  co-prime and
\begin{equation}
0\leq N< \mbox{lcm}\left(M_1, M_2, \cdots, M_L\right)=M\Gamma_1\Gamma_2\cdots\Gamma_L.
\end{equation}
If the remainder error bound $\tau$ satisfies
\begin{equation}
\tau< M/4,
\end{equation}
then we have $\hat{n}_i=n_i$ for $1\leq i \leq L$, i.e.,
the folding numbers $n_i$ for $1\leq i \leq L$ can be accurately
determined.
\end{proposition}

As it was mentioned earlier,
these robust reconstruction results
are
based on
the assumption that the gcd $M$
 of all the moduli is more than $1$ and the remaining integers
$\Gamma_i$ in the moduli $M_i$
factorized by their gcd $M$
are  co-prime. For example,  $M_1=5\cdot 5=25$, $M_2=5\cdot 7=35$,
$M_3=5\cdot 16=80$, and $M_4=5\cdot 19=95$, where $M=5$. When the remainder error
level $\tau<5/4$, any integer less than $5\cdot 5 \cdot 7 \cdot 16\cdot 19$ can be reconstructed within the same error level as the remainders
from the erroneous remainders by using the algorithms
in \cite{xwli2, wjwang}.
A natural question is what will happen if a general set of
moduli $M_i$ are used. For example, what will happen if
$M_1=5\cdot 14=70$, $M_2=5\cdot 15=75$,
$M_3=5\cdot 16=80$,
and $M_4=5\cdot 18=90$? First of all, their gcd is $M=5$ and if we
divide them by their gcd, we get $\Gamma_1=14$, $\Gamma_2=15$,
$\Gamma_3=16$, and $\Gamma_4=18$
and clearly these four  $\Gamma_i$ are not co-prime.
So, we can not apply the  algorithms or results in \cite{xwli2, wjwang}
directly, which may limit the applications in practice.

We next propose an accurate determination algorithm for the folding
numbers $n_i$ from erroneous remainders for a general set of moduli
$M_i$ with a new necessary and sufficient condition on the remainder
errors. Let us first see an algorithm for $n_i$.


Following the algorithm  in \cite{wjwang}, we can  generalize the results
as follows.
First, from (\ref{CRT}) we can equivalently write it as the following system of congruences:
\begin{equation}\label{fangchengzu}
\left\{\begin{array}{ll}
N=n_1M_1+r_1\\
N=n_2M_2+r_2\\
\:\:\:\:\vdots\\
N=n_LM_L+r_L.
\end{array}\right.
\end{equation}
We want to determine $n_i$ for $1\leq i \leq L$. To do so, we let the last $L-1$ equations in (\ref{fangchengzu}) subtract the first one and we then have
\begin{equation}\label{fangchengzu2}
\left\{\begin{array}{ll}
n_1M_1-n_2M_2=r_2-r_1\\
n_1M_1-n_3M_3=r_3-r_1\\
\:\:\:\:\vdots\\
n_1M_1-n_LM_L=r_L-r_1.
\end{array}\right.
\end{equation}
Next, denote
$$m_{1i}=\mbox{gcd}\left(M_1, M_i\right),\,\,\,
\Gamma_{1i}=\frac{M_1}{m_{1i}}, \,\,\, \Gamma_{i1}=\frac{M_i}{m_{1i}},
\mbox{ and } q_{i1}=\frac{r_i-r_1}{m_{1i}}.$$
 Then, we can equivalently express equation (\ref{fangchengzu2}) again as
\begin{equation}\label{fangchengzu3}
\left\{\begin{array}{ll}
n_1\Gamma_{12}-n_2\Gamma_{21}=q_{21}\\
n_1\Gamma_{13}-n_3\Gamma_{31}=q_{31}\\
\:\:\:\:\vdots\\
n_1\Gamma_{1L}-n_L\Gamma_{L1}=q_{L1}.
\end{array}\right.
\end{equation}
Since $\Gamma_{1i}$ and $\Gamma_{i1}$ are co-prime, by B\'{e}zout's lemma
(Lemma 1 in \cite{wjwang}) we get
\begin{equation}\label{fangchengzu4}
\left\{\begin{array}{ll}
n_1=q_{i1}\overline{\Gamma}_i+kq_{i1}\Gamma_{i1}\\
n_i=\frac{q_{i1}(\overline{\Gamma}_i\Gamma_{1i}-1)}{\Gamma_{i1}}+kq_{i1}\Gamma_{1i},
\end{array}\right.
\end{equation}
where $1\leq i \leq L$, $k\in \mathbb{Z}$ (the set of integers) and $\overline{\Gamma}_i$ is the modular multiplicative inverse of $\Gamma_{1i}$ modulo $\Gamma_{i1}$.

We can use
\begin{equation}\label{qi1}
\hat{q}_{i1}=\left[\frac{\tilde{r}_i-\tilde{r}_1}{m_{1i}}\right]=q_{i1}+\left[\frac{\Delta r_i-\Delta r_1}{m_{1i}}\right]
\end{equation}
as an estimate of $q_{i1}$
for $2\leq i\leq L$.
 Recall that $\left[\cdot\right]$ stands for the rounding integer which is defined in (\ref{round integer}).
Let $\hat{n}_{i}$ for $1\leq i\leq L$ be a set of solutions of (\ref{fangchengzu3}) when $q_{i1}$ is replaced by $\hat{q}_{i1}$ for $2\leq i\leq L$.
In summary, we have the following algorithm.

\begin{itemize}
  \item \textbf{Step 1:} Calculate these values of $m_{1i}=\mbox{gcd}\left(M_1, M_i\right)$, $\Gamma_{1i}=\frac{M_1}{m_{1i}}$ and $\Gamma_{i1}=\frac{M_i}{m_{1i}}$ for $2 \leq i\leq L$ from the given moduli $M_j$ for $1 \leq j\leq L$, which
can be done in advance.
  \item \textbf{Step 2:} Calculate $\hat{q}_{i1}$ for $2 \leq i\leq L$ in (\ref{qi1}) from the given erroneous remainders $\tilde{r}_i$ for $1 \leq i\leq L$.
  \item \textbf{Step 3:} Calculate the remainders of $\hat{q}_{i1}\overline{\Gamma}_i$ modulo $\Gamma_{i1}$, i.e.,
  \begin{equation}
   \hat{\xi}_{i1} \equiv \hat{q}_{i1}\overline{\Gamma}_i \mbox{ mod } \Gamma_{i1}
  \end{equation}
  for $2 \leq i\leq L$, where $\overline{\Gamma}_i$ is the modular multiplicative inverse of $\Gamma_{1i}$ modulo $\Gamma_{i1}$ and  can be calculated in advance.
  \item \textbf{Step 4:} Calculate $\hat{n}_{1}$ from the following system of congruences:
  \begin{equation}\label{gcrt}
  \hat{n}_1\equiv \hat{\xi}_{i1} \mbox{ mod } \Gamma_{i1}, \mbox{ for } 2 \leq i\leq L,
  \end{equation}
  where moduli $\Gamma_{i1}$ may not be co-prime,
 which can be done by using the algorithms in, for example, \cite{ore, newcrt}, and in
\cite{newcrt},  a multi-level  decoding technique to reconstruct the large integer is proposed.
  \item \textbf{Step 5:} Calculate $\hat{n}_{i}$ for $2 \leq i\leq L$:
  \begin{equation}\label{folding}
  \hat{n}_{i}=\frac{\hat{n}_{1}\Gamma_{1i}-\hat{q}_{i1}}{\Gamma_{i1}}.
  \end{equation}
\end{itemize}
With the above  algorithm,  we have the following necessary and sufficient
condition result
for a general set of moduli.
\begin{theorem}\label{them1}
Let $M_i$, $1\leq i \leq L$,
be  $L$ arbitrarily distinct positive integers as a given
set of moduli
 and $0\leq N< \mbox{lcm}\left(M_1, M_2, \cdots, M_L\right)$. Then,  $\hat{n}_i=n_i$ for all $1\leq i \leq L$, i.e.,
the folding numbers $n_i$ for $1\leq i \leq L$ can be accurately
determined,
 if and only if
\begin{equation}\label{condition1}
-\mbox{gcd} \left(M_1, M_i\right)/2\leq \Delta r_i-\Delta r_1< \mbox{gcd} \left(M_1, M_i\right)/2, \:\:\:\:\mbox{for all} \:\:\:\:2\leq i \leq L.
\end{equation}
\end{theorem}
\begin{proof}
We first prove the sufficiency. Considering the condition in (\ref{condition1}) and the estimate of $q_{i1}$ in (\ref{qi1}),
from (\ref{round integer}) for the definition of the operator $[\cdot]$
we have $\hat{q}_{i1}=q_{i1}$. Then, from equation (\ref{fangchengzu4}), $n_1$ and $\hat{q}_{i1}\overline{\Gamma}_i$ have the same remainder modulo $\Gamma_{i1}$. Since $\hat{q}_{i1}$, $\overline{\Gamma}_i$ and $\Gamma_{i1}$
are  known, we can calculate $\hat{q}_{i1}\overline{\Gamma}_i \equiv \hat{\xi}_{i1} \mbox{ mod } \Gamma_{i1}$. Thus, $n_1\equiv \hat{\xi}_{i1} \mbox{ mod } \Gamma_{i1}$ for $2\leq i\leq L$, which form a system of simultaneous congruences as $\hat{n}_1\equiv \hat{\xi}_{i1} \mbox{ mod } \Gamma_{i1}$. In addition,
since $n_1M_1\leq N<$lcm$(M_1,M_2,\cdots,M_L)$, it is not hard to see that
$n_1$ is  less than $\mbox{lcm} \left( \Gamma_{21}, \Gamma_{31}, \cdots, \Gamma_{L1}\right)$. So, according to the algorithm about generalized CRT in \cite{newcrt}, $n_1$ can be uniquely reconstructed by solving the above system, and $n_1=\hat{n}_{1}$.

After $n_1$ is determined, we can obtain other integers $n_i$ for $2\leq i\leq L$ from equations (\ref{fangchengzu3}) or (\ref{fangchengzu4}). Therefore, $\hat{n}_i=n_i$ for $2\leq i\leq L$. Hence, the sufficiency is proved.

We next prove the necessity. Assume that there exists at least one remainder that does not satisfy (\ref{condition1}). For example, the $j$-th remainder
$\tilde{r}_j$, $2\leq j\leq L$,
does not satisfy (\ref{condition1}).
This equivalently leads to
 $\left[\left(\Delta r_j-\Delta r_1\right)/m_{1j}\right]\neq 0$ and
therefore $\hat{q}_{j1}\neq q_{j1}$. We then have the following two cases.\\
\textbf{Case A:} When $\left[\left(\Delta r_j-\Delta r_1\right)/m_{1j}\right]\neq k\Gamma_{j1}$ for any $k\in \mathbb{Z}$. We want to prove that the remainders of $\hat{q}_{j1}\overline{\Gamma}_j$ and $q_{j1}\overline{\Gamma}_j$ modulo $\Gamma_{j1}$ are different. Assume $\hat{q}_{j1}\overline{\Gamma}_j$ and $q_{j1}\overline{\Gamma}_j$ have the same remainder modulo $\Gamma_{j1}$, i.e.,
\begin{equation}\label{aaaaa}
\hat{q}_{j1}\overline{\Gamma}_j-q_{j1}\overline{\Gamma}_j=k\Gamma_{j1}, \:\:\mbox{for some}\:\: k \in \mathbb{Z}.
\end{equation}
Multiplying both sides of (\ref{aaaaa}) by $\Gamma_{1j}$ and considering $\Gamma_{1j}\overline{\Gamma}_j=1+k\Gamma_{j1}$ for some $k \in \mathbb{Z}$, we have
\begin{equation}
\hat{q}_{j1}-q_{j1}=k\Gamma_{j1}, \:\:\mbox{for some}\:\: k \in \mathbb{Z}.
\end{equation}
According to (\ref{qi1}), we have
\begin{equation}
\left[\left(\Delta r_j-\Delta r_1\right)/m_{1j}\right]=k\Gamma_{j1}, \:\:\mbox{for some}\:\: k \in \mathbb{Z}.
\end{equation}
This contradicts with the assumption. Hence, the remainders of $\hat{q}_{j1}\overline{\Gamma}_j$ and $q_{j1}\overline{\Gamma}_j$ modulo $\Gamma_{j1}$ are different, i.e.,  $n_1$ and $\hat{n}_1$ have different congruences. Thus, $n_1\neq \hat{n}_1$.\\
\textbf{Case B:} For every $2\leq i\leq L$, $\left[\left(\Delta r_i-\Delta r_1\right)/m_{1i}\right]=k\Gamma_{i1}$ for some $k \in \mathbb{Z}$ but there exists at least one $j$ with $2\leq j\leq L$ such that $\left[\left(\Delta r_j-\Delta r_1\right)/m_{1j}\right]\neq 0$,  i.e., $\hat{q}_{j1}\neq q_{j1}$.
From equation (\ref{qi1}), we have $\hat{q}_{i1}\overline{\Gamma}_i \equiv q_{i1}\overline{\Gamma}_i \mbox{ mod } \Gamma_{i1}$ for $2\leq i\leq L$.
Hence, from the first equation in (\ref{fangchengzu4})
and  according to the generalized CRT, $n_1$ can be uniquely reconstructed.
Thus, from Steps 1-4 in the above algorithm,
 we have $\hat{n}_1=n_1$.

However, since $\hat{q}_{j1}\neq q_{j1}$, from equations (\ref{fangchengzu3}) or the second equation
in (\ref{fangchengzu4}) we have $n_j\neq \hat{n}_j$. This proves the necessity.
\end{proof}

The above result involves with two remainder errors and is hard to check in practice. Let $\tau$ be the maximal remainder error level, i.e., $|\Delta r_i|=|\tilde{r}_i-r_i|\leq \tau$, for $1\leq i \leq L$. Similar to Proposition \ref{wjwcorollary}, we can also present a simpler sufficient condition.
\begin{corollary}\label{111}
Let $M_i$, $1\leq i \leq L$,
be  $L$ arbitrarily distinct positive integers as a given
set of moduli
 and $0\leq N< \mbox{lcm}\left(M_1, M_2, \cdots, M_L\right)$.
If the remainder error bound $\tau$ satisfies
\begin{equation}\label{tttttt}
\tau< \max_{1\leq i\leq L}\min_{1\leq j\neq i\leq L} \frac{\mbox{gcd}\left(M_i, M_j\right)}{4},
\end{equation}
then, we have $\hat{n}_i=n_i$ for all $1\leq i \leq L$, i.e.,
the folding numbers $n_i$ for $1\leq i \leq L$ can be accurately
determined.
\end{corollary}
\begin{proof}
Recall that in the procedure of proving Theorem \ref{them1} we just arbitrarily selected the first equation in (\ref{fangchengzu}) to be a reference to be subtracted from the  other equations to
get (\ref{fangchengzu2}). In fact, to improve the robustness through selecting a proper reference equation to differentiate, we can choose the index $i$ such that $\min\limits_j \mbox{gcd}\left(M_i, M_j\right)\geq \min\limits_j \mbox{gcd}\left(M_l, M_j\right)$ for $l\neq i$. Without loss of generality, modulus $M_1$ satisfies $$\max\limits_{1\leq i\leq L}\min\limits_{1\leq j\neq i\leq L} \frac{\mbox{gcd}\left(M_i, M_j\right)}{4}=\min\limits_{1\leq j\neq 1\leq L} \frac{\mbox{gcd}\left(M_1, M_j\right)}{4}.$$ Then, we have $\tau<\mbox{gcd}(M_1, M_i)/4$ for $i\neq1$.
Since $\tau$ is the maximal remainder error level, i.e., $|\Delta r_i|=|\tilde{r}_i-r_i|\leq \tau$, for $1\leq i \leq L$, we can obtain
\begin{equation}\label{bbbb}
|\Delta r_i-\Delta r_1|<\mbox{gcd}\left(M_1, M_i\right)/2, \mbox{ for } 2\leq i \leq L.
\end{equation}
Clearly, equation (\ref{bbbb}) implies the sufficient condition (\ref{condition1}) in Theorem \ref{them1}. Hence, $\hat{n}_i=n_i$ for all $1\leq i \leq L$. Therefore, we complete the proof.
\end{proof}

\begin{remark}
Since in the above new result,
 there is no any constraint to the moduli $M_i$ for  $1\leq i \leq L$,
some of the moduli may be  redundant with respect to the
range $0\leq N< \mbox{lcm}\left(M_1, M_2, \cdots, M_L\right)$
of the determinable unknown integer $N$.
The first case is when there exist a pair of moduli $M_{i_1}$ and $M_{i_2}$ such that $M_{i_1}=nM_{i_2}$ for $n \in \mathbb{N}$ (the set of all positive integers) and in this case
$M_{i_2}$ is redundant for the determinable range of $N$, i.e., the lcm of
all  $M_i$.
 The other case is when there exists one moduli $M_{i_3}$
that  is a factor of some other (more than one)
moduli's lcm, i.e., $\mbox{lcm}\left(M_{i_4}, M_{i_5}, \cdots, M_{i_k}\right)=nM_{i_3}$ for some $n \in \mathbb{N}$ and $k>4$,
and in this case $M_{i_3}$ is redundant similarly.
When a determinable range of $N$ is fixed,
we  can add or delete some of the redundant moduli to or from the moduli set
 in order to get a better robustness bound for $\tau$.
 For example, the redundant modulus $30$ in moduli set
$\{20, 45, 30\}$ improves the robustness bound
compared with the robustness bound of
moduli set $\{20, 45\}$ from $5/4$ to $10/4$.
On the other hand, the redundant modulus $10$ in $\{10, 45, 30\}$
with its robustness bound $10/4$ does not
help but  worsens the robustness bound compared with
$15/4$ of $\{45, 30\}$, so it is better to delete the modulus $10$ from the moduli set. Below is a general result.
\end{remark}

\begin{corollary}\label{cor2}
If there exist a pair of moduli $M_{i_1}$ and $M_{i_2}$ such that $M_{i_1}=nM_{i_2}$ for $n \in \mathbb{N}$, then,
the redundant modulus $M_{i_2}$ does not help to increase
 the robustness bound and it can be deleted from the set of moduli.
\end{corollary}
\begin{proof}
Without loss of generality, we can assume for a moduli set $\{M_1, M_2, \cdots, M_L\}$ that $\mbox{gcd}(M_1, M_2)\geq \mbox{gcd}(M_1, M_3)\geq \cdots\geq\mbox{gcd}(M_1, M_L)$ and the robustness bound is
$$\max\limits_{1\leq i\leq L}\min\limits_{1\leq j\neq i\leq L} \frac{\mbox{gcd}\left(M_i, M_j\right)}{4}=\frac{\mbox{gcd}(M_1, M_L)}{4}.$$

Consider another set of moduli
$\{M_1, M_2, \cdots, M_L, M_{L+1}\}$ where $M_{L+1}$ is
a factor of one moduli $M_q$ in $\{M_1, M_2, \cdots, M_L\}$, i.e., $M_{q}=nM_{L+1}$ for $n \in \mathbb{N}$, $1\leq q\leq L$.
For the moduli $\{M_1, M_2, \cdots, M_L, M_{L+1}\}$,
 its robustness bound is
$$\max\limits_{1\leq i\leq L+1}\min\limits_{1\leq j\neq i\leq L+1} \frac{\mbox{gcd}\left(M_i, M_j\right)}{4}.$$
To calculate it, we  split $1\leq i\leq L+1$ into two parts: $1\leq i\leq L$ and $i=L+1$. And,
$$
\max\limits_{1\leq i\leq L+1}\min\limits_{1\leq j\neq i\leq L+1} \frac{\mbox{gcd}\left(M_i, M_j\right)}{4}=\max\{\underbrace{\max\limits_{1\leq i\leq L}\min\limits_{1\leq j\neq i\leq L+1} \frac{\mbox{gcd}\left(M_i, M_j\right)}{4}}_A, \underbrace{\min\limits_{1\leq j\leq L} \frac{\mbox{gcd}\left(M_{L+1}, M_j\right)}{4}}_B\}.
$$
As for $A$, since
$$\min\limits_{\begin{subarray}{}
                                  1\leq j\neq i\leq L+1 \\
                                  1\leq i\leq L
                                \end{subarray}
} \frac{\mbox{gcd}\left(M_i, M_j\right)}{4}\leq \min\limits_{\begin{subarray}{}
                                  1\leq j\neq i\leq L \\
                                  1\leq i\leq L
                                \end{subarray}
} \frac{\mbox{gcd}\left(M_i, M_j\right)}{4},$$
we have
$$
A\leq \max\limits_{1\leq i\leq L}\min\limits_{1\leq j\neq i\leq L} \frac{\mbox{gcd}\left(M_i, M_j\right)}{4}=\frac{\mbox{gcd}(M_1, M_L)}{4}.
$$
As for $B$, since
$$\mbox{gcd}(M_{L+1}, M_j)\leq \mbox{gcd}(M_q, M_j) \mbox{ for }1\leq j\leq L,$$
 we have
$$
B\leq \min\limits_{1\leq j\neq q\leq L} \frac{\mbox{gcd}\left(M_{L+1}, M_j\right)}{4} \leq \max\limits_{1\leq i\leq L}\min\limits_{1\leq j\neq i\leq L} \frac{\mbox{gcd}\left(M_i, M_j\right)}{4}=\frac{\mbox{gcd}(M_1, M_L)}{4}.
$$
Thus, we can derive
$$\max\limits_{1\leq i\leq L+1}\min\limits_{1\leq j\neq i\leq L+1} \frac{\mbox{gcd}\left(M_i, M_j\right)}{4}=\max\{A, B\}\leq \max\limits_{1\leq i\leq L}\min\limits_{1\leq j\neq i\leq L} \frac{\mbox{gcd}\left(M_i, M_j\right)}{4}=\frac{\mbox{gcd}(M_1, M_L)}{4}.
$$
This tells us that the redundant modulus $M_{L+1}$ does not help
to increase  the robustness bound of the set of moduli $\{M_1, M_2, \cdots, M_L, M_{L+1}\}$ compared with that of $\{M_1, M_2, \cdots, M_L\}$.
\end{proof}

From the result of Corollary \ref{cor2}, for a set of moduli, we can delete this kind of redundant modulus $M_{i_2}$ when there exists one modulus $M_{i_1}$ in the moduli set such that $M_{i_1}=nM_{i_2}$. So, throughout this paper, a set of moduli we consider does not include such a pair of moduli in a single
stage robust CRT.

From the above results, one can see that
the choice of the reference remainder is important in determining
the maximal possible robustness bound for $\tau$ when
the whole moduli set of $L$ arbitrary moduli is considered once as above.
In fact, when the moduli satisfy the constraint, i.e.,
$\Gamma_i$ are co-prime,
in Proposition \ref{wjwcorollary} in \cite{wjwang,xwli2},
it has been pointed out and analyzed in \cite{wjwang} that
a proper reference remainder indeed plays an important role
in improving the performance in practice.

Going back to the necessary and sufficient condition (\ref{condition1}),
one can see that  the remainder error
difference bound depends on $\mbox{gcd}(M_1, M_i)$,
that  varies with each $M_i$, and the  choice of the reference modulus
$M_1$.
This means that for the robust reconstruction of $N$,
the error levels of its remainders $\tilde{r}_i$ for different $i$ may have
 different requirements. Also, as it was mentioned earlier,
using $M_1$ as the reference modulus is not necessary.
Let us choose the reference modulus $M_k$ that satisfies
 \begin{equation}\label{refer}
\min_{1\leq j\neq k\leq L} \frac{\mbox{gcd}\left(M_k, M_j\right)}{4}=\max_{1\leq i\leq L}\min_{1\leq j\neq i\leq L} \frac{\mbox{gcd}\left(M_i, M_j\right)}{4},
 \end{equation}
 and the remainder error bound $\tau_k$ for the reference remainder
$r_k$ satisfy
 \begin{equation}\label{refer1}
 |\Delta r_k|=|\tilde{r}_k-r_k|\leq \tau_k<\min\limits_{1\leq j\neq k\leq L} \frac{\mbox{gcd}\left(M_k, M_j\right)}{4}.
 \end{equation}
Then, we have the following result.
\begin{corollary}\label{coroll2}
Let $M_i$, $1\leq i \leq L$,
be  $L$ arbitrarily distinct positive integers as a given
set of moduli
 and $0\leq N< \mbox{lcm}\left(M_1, M_2, \cdots, M_L\right)$, define the remainder error bound for $r_i$ as $\tau_i$, i.e., $|\Delta r_i|=|\tilde{r}_i-r_i|\leq \tau_i$ for $1\leq i\leq L$, and the reference modulus and its corresponding remainder error bound are $M_k$ and $\tau_k$ satisfying (\ref{refer})
and (\ref{refer1}) above for
some $k$ with $1\leq k\leq L$.
If the remainder error bound $\tau_i$, $1\leq i\neq k\leq L$, satisfies
\begin{equation}\label{tttttt2222}
|\Delta r_i|=|\tilde{r}_i-r_i|\leq \tau_i\leq\frac{\mbox{gcd}\left(M_k, M_i\right)}{2}-\min\limits_{1\leq j\neq k\leq L} \frac{\mbox{gcd}\left(M_k, M_j\right)}{4},
\end{equation}
then, we have $\hat{n}_i=n_i$ for all $1\leq i \leq L$, i.e.,
the folding numbers $n_i$ for $1\leq i \leq L$ can be accurately
determined.
 \end{corollary}

\begin{proof}
If the reference modulus and its corresponding remainder error bound are $M_k$ and $\tau_k$, $1\leq k\leq L$, we just need to prove $|\Delta r_i-\Delta r_k|<\mbox{gcd}(M_i, M_k)/2$ for all $1\leq i\neq k\leq L$, which implies the sufficient condition (\ref{condition1}) in Theorem \ref{them1}.

Since $|\Delta r_i|\leq \tau_i\leq\frac{\mbox{gcd}\left(M_k, M_i\right)}{2}-\min\limits_{1\leq j\neq k\leq L} \frac{\mbox{gcd}\left(M_k, M_j\right)}{4}$, we have
\begin{equation}
|\Delta r_i-\Delta r_k|\leq |\Delta r_i|+|\Delta r_k|\leq \tau_i+\tau_k<\frac{\mbox{gcd}\left(M_k, M_i\right)}{2}.
\end{equation}
Thus, Corollary  \ref{coroll2} is proved.
\end{proof}

Next, we consider the example mentioned before again.
\begin{example}
Let $M_1=5\cdot 14=70$, $M_2=5\cdot 15=75$, $M_3=5\cdot 16=80$, and $M_4=5\cdot 18=90$. It is easy to see that their gcd is $M=5$, and $\{M_i\}_{i=1,2,3,4}$ do not satisfy the constraint of Propositions \ref{wjwtheorem} and
\ref{wjwcorollary} in \cite{wjwang, xwli2}. Thus, their results can
 not be applied here. However,
 from the result of Corollary \ref{111}, we can obtain the maximal robustness bound
$\tau$ for all remainders as  $10/4$, which  is even  larger than $5/4$, a quarter of the gcd of all the moduli.
 From the result of Corollary \ref{coroll2}, we choose $M_4$ as the reference modulus that does satisfy
 (\ref{refer}), and we can get the robustness bound for each remainder as follows: $\tau_1\leq10/4$, $\tau_2\leq20/4$, $\tau_3\leq 10/4$ and $\tau_4< 10/4$.
One can see that,
if we treat remainder error bounds individually
as above, the remainder error bounds for some of the individual remainders,
such as the second remainder in  this example,
may be larger than that in (\ref{tttttt}) in Corollary \ref{111}
 for all the remainder error levels.
In addition, the robust reconstruction range of $N$ is also $0\leq N< $ lcm$(M_1, M_2, M_3, M_4)$.
\end{example}

It is clear that  when  moduli $\Gamma_{i1}$ for $2\leq i\leq L$ are co-prime similar to the case of \cite{wjwang},
from the system of congruences (\ref{gcrt}) in Step 4 in the above algorithm,
 a closed-form single stage CRT can be obtained as \cite{wjwang}
and we can replace Step 4 with the following
Step $4^{\star}$.
\begin{itemize}
  \item \textbf{Step $4^{\star}$:} Calculate $\hat{n}_1$:
\end{itemize}
\begin{equation}\label{close}
\hat{n}_1=\sum_{i=2}^{L}\hat{\xi}_{i1}b_{i1}\frac{\gamma}{\Gamma_{i1}} \mbox{ mod } \gamma,
\end{equation}
where $b_{i1}$ is the modular multiplicative inverse of $\gamma/\Gamma_{i1}$ modulo $\Gamma_{i1}$, which can be calculated in advance, and $\gamma=\Gamma_{21}\Gamma_{31}\cdots\Gamma_{L1}$. After that, from (\ref{folding}) we can get the formulas for other $\hat{n}_i$ for $2\leq i\leq L$.

Next, let us consider the result in \cite{binyang}.
If we consider the following
special case of moduli in Corollary \ref{coroll2},  we can obtain a better result of the remainder error bounds than that in \cite{binyang}.
Let a set of moduli be $$\{M_1, M_2, \cdots, M_{L_1+L_2}\}=\{M\Gamma_{1,1}, M\Gamma_{1,2}, \cdots, M\Gamma_{1,L_1}, MM^{'}\Gamma_{2,1},  MM^{'}\Gamma_{2,2}, \cdots, MM^{'}\Gamma_{2,L_2}\},$$
where $L_1\geq2$. Then, we have the following corollary.

\begin{corollary}\label{cor3}
Assume that all the $\Gamma_{j, i}$, $1\leq i\leq L_j$, $j=1, 2$, are pair-wisely co-prime, let $\Gamma_{1,1}=\dot{\Gamma}_{1,1}M^{'}$, where $\dot{\Gamma}_{1,1}$ is an integer, and  $0\leq N<\mbox{lcm}(M_1, M_2, \cdots, M_{L_1+L_2})$.
Denote $\tau_i$ as the error bound for each remainder $r_i$ for $1\leq i\leq L_1+L_2$. If
\begin{equation}
\begin{array}{cc}
 |\Delta r_1|\leq \tau_1< M/4,\\
 |\Delta r_i|\leq \tau_i\leq M/4, \:\:\:\:\mbox{for all} \:\:\:\:2\leq i \leq L_1, \\
 |\Delta r_i|\leq \tau_i\leq MM^{'}/2-M/4, \:\:\:\:\mbox{for all} \:\:\:\:L_1+1\leq i \leq L_1+L_2,
\end{array}
\end{equation}
then with a closed-form algorithm we have $\hat{n}_i=n_i$ for $1\leq i\leq L_1+L_2$, i.e., the folding numbers $n_i$ for $1\leq i\leq L_1+L_2$ can be accurately determined.
\end{corollary}

\begin{proof}
Since $$\min\limits_{1< j\leq L_1+L_2} \frac{\mbox{gcd}\left(M_1, M_j\right)}{4}=\max\limits_{1\leq i\leq L_1+L_2}\min\limits_{1\leq j\neq i\leq L_1+L_2} \frac{\mbox{gcd}\left(M_i, M_j\right)}{4},$$ we can set $M_1$ as the reference modulus and
the error bound $\tau_1<\min\limits_{1< j\leq L_1+L_2} \frac{\mbox{gcd}\left(M_1, M_j\right)}{4}=\frac{M}{4}$ in (\ref{refer1}).

Then, from (\ref{tttttt2222}) when $2\leq i \leq L_1$, $\tau_i\leq M/2-M/4=M/4$; and when $L_1+1\leq i \leq L_1+L_2$, $\tau_i<MM^{'}/2-M/4$. So we can accurately determine the folding numbers $n_i$ for $1\leq i\leq L_1+L_2$.
Next, we can get $\Gamma_{i1}=\frac{M_i}{\mbox{gcd}(M_1, M_i)}=\Gamma_{1,i}$ for $2\leq i\leq L_1$ and $\Gamma_{i1}=\Gamma_{2,i}$ for $L_1+1\leq i \leq L_1+L_2$, all of which are co-prime. Thus, we can obtain a simple
closed-form reconstruction formula for $\hat{n}_1$ similar to  (\ref{close})
and then  $\hat{n}_i$ by (\ref{folding}) for $2\leq i\leq L_1+L_2$.
\end{proof}

\begin{example}
In the above, let $M=10, M^{'}=3, \Gamma_{1,1}=33, \Gamma_{1,2}=31, \Gamma_{2,1}=35, \Gamma_{2,2}=37$, and the moduli are $\{M_1, M_2, M_3, M_4\}=\{10\cdot33, 10\cdot31, 30\cdot35, 30\cdot37\}$. From Corollary \ref{cor3}, we can get $\tau_1<10/4$, $\tau_2\leq 10/4$, $\tau_3\leq 50/4$, $\tau_4\leq 50/4$.
In addition, it has  a closed-form algorithm to robustly reconstruct an unknown integer $N$ for $0\leq N<\mbox{lcm}(M_1, M_2, M_3, M_4)$. However, according to the result of \cite{binyang}, the remainder error bounds would be $\tau_1<10/4$, $\tau_2< 10/4$, $\tau_3< 30/4$, $\tau_4< 30/4$.
\end{example}

Interestingly, the robustness bound result in this corollary
is even better than
that obtained in \cite{binyang}
using a two-stage robust CRT.
What the result here tells us that for the set of moduli in Corollary
\ref{cor3}, which is the set considered in \cite{binyang},
it is not necessary to use a two-stage robust CRT as what is done in
\cite{binyang}.
Another remark we make here is
that the notation $\tau_i$ above denotes
the $i$th remainder error bound. Later, without causing any
notational  confusion, $\tau_j$ will denote the remainder error bound
for the remainders in the j-th group.

\section{Multi-Stage Robust CRT}\label{sec3}

From the study  in the previous section, one
can see that the robustness bound is kind of dependent
on the gcd of the moduli. The larger the gcd is, the better the robustness bound is. However, the large gcd reduces the lcm of the moduli, i.e., reduces
the determinable range of the unknown integer $N$.
When a set of moduli are given, the maximal determinable range
is given too, which is their lcm.
Then, the question is for a given set of moduli, can we improve the robustness bound obtained in Corollary \ref{111}?
Note that in the single stage robust CRT obtained in the previous section,
all the remainders
and their related system of congruence equations are considered and solved
 together simultaneously.
A natural question is: can we split the set of moduli into several
groups so that the moduli in each group have a large gcd and remainders
and their corresponding system of congruence equations in each group
are considered and solved independently using the single stage robust CRT
obtained in the previous section? If so, can we obtain a better robustness
bound than that in Corollary \ref{111} for the single stage robust CRT?
To answer these questions, let us first see an example.



Suppose that a set of moduli $\{180, 220, 486, 513\}$ are given and
 the gcd of these $4$ moduli is $1$.
These four moduli can be
split to two groups $\{180, 220\}$ and $\{486, 513\}$. The gcd of the two moduli in the first group is $M=20$ with
$\Gamma_1=180/20=9$, $\Gamma_2=220/20=11$,  and the gcd of the
two moduli in the second group is $M=27$ with $\Gamma_1=486/27=18$, $\Gamma_2=513/27=19$.
One can see that each group satisfies the condition in
Propositions \ref{wjwtheorem} and \ref{wjwcorollary} and
therefore the closed-form
robust CRT in \cite{wjwang} or the single stage robust CRT
in the previous section can be applied for the robust reconstruction
of an unknown integer
$N$ with robustness bound $\tau < 20/4$ or $\tau <  27/4$,
if $N<$ lcm$\{180, 220\}=1980\stackrel{\Delta}{=}q_1=18\cdot 110$ or $N<$ lcm$\{486, 513\}=9234\stackrel{\Delta}{=}q_2=18\cdot 513$, respectively.

Using the first group with moduli $M_1=180$ and $M_2=220$ and two remainders
$r_1$ and $r_2$, if the integer $N$ is in the range of $[0, \mbox{lcm}(M_1, M_2))$,
then $N$ can be uniquely determined by its two error free remainders $r_1$ and $r_2$ as $N_1$ with $0\leq N_1<q_1$;
otherwise
\begin{equation}\label{g1}
N= N_1 \mbox{ mod }q_1.
\end{equation}
Using two erroneous remainders $\tilde{r}_1$ and $\tilde{r}_2$ with error level $\tau$, and
the closed-form
robust CRT in \cite{wjwang} or the single stage robust CRT
in the previous section
 for the first group, we can obtain an integer $\hat{N}_1$ 
and
\begin{equation}\label{g11}
\mbox{if }\tau < \mbox{ gcd}(M_1, M_2)/4=20/4,
\mbox{ then } |N_1-\hat{N}_1|\leq \tau.
\end{equation}

Similarly, using the second group with moduli $M_3=486$ and $M_4=513$ and two remainders
$r_3$ and $r_4$, if the integer $N$ is in the range of $[0, \mbox{lcm}(M_3, M_4))$,
then $N$ can be uniquely determined by its  two error free remainders $r_3$ and $r_4$ as $N_2$ with $0\leq N_2<q_2$;
otherwise
\begin{equation}\label{g2}
N= N_2 \mbox{ mod }q_2.
\end{equation}
Using two erroneous remainders $\tilde{r}_3$ and $\tilde{r}_4$ with error level $\tau$, and the
 closed-form
robust CRT in \cite{wjwang} or the single stage robust CRT
in the previous section
for the second group, we can obtain an integer $\hat{N}_2$ 
and
\begin{equation}\label{g21}
\mbox{if }\tau < \mbox{ gcd}(M_3, M_4)/4=27/4,
\mbox{ then } |N_2-\hat{N}_2|\leq \tau.
\end{equation}

On the other hand, if integer $N$ is in the range of
$[0, \mbox{ lcm}(M_1, M_2, M_3, M_4))=[0, \mbox{ lcm}(q_1, q_2))$, it can be uniquely determined by its four error free remainders $r_1, r_2, r_3, r_4$. This can be done either from the four remainders directly or from the two new
remainders $N_1$ and $N_2$ of $N$ with two new moduli $q_1$ and $q_2$
with equations (\ref{g1}) and (\ref{g2}), respectively.
For the robustness, as we mentioned earlier,
the closed-form robust CRT and the results in Propositions
\ref{wjwtheorem} and \ref{wjwcorollary}
can not be applied to the four moduli and  the four erroneous remainders
directly since they do not satisfy the co-prime
condition. In addition, according to our single stage robust CRT in Theorem \ref{them1} and Corollary \ref{111} obtained
in the previous section,
its robustness bound would be $\tau< 9/4$ (interestingly,
for the $4$ moduli, their gcd is only $1$). However, using the above
grouping idea, the reconstruction of $N$
 can be done in two stages: the first stage is to reconstruct
$\hat{N}_1$ in (\ref{g11}) and $\hat{N}_2$ in (\ref{g21}) from the two groups,
respectively; the second stage is to reconstruct $\hat{N}$ from its two possibly
erroneous remainders $\hat{N}_1$ and $\hat{N}_2$
with two new moduli $q_1$ and $q_2$.
From the second stage, using the known robust CRT again, we obtain
\begin{equation}\label{gg}
\mbox{if }\tau<\mbox{ gcd}(q_1, q_2)/4=18/4,\mbox{ then, }
|N-\hat{N}|\leq \tau.
\end{equation}
Thus, we have a robust reconstruction too. In order to keep
 all inequalities (\ref{g11}), (\ref{g21}) and (\ref{gg}), one can see that with this two-stage
approach,
the robustness bound on the remainder error level $\tau$ is $18/4$
which is surprisingly even better
than $9/4$ that is the robustness
bound in Corollary \ref{111} using  the single stage
robust CRT for general moduli obtained in the previous section.
This means that using two or more groups for a set of moduli
may have a  better robustness
bound than that using a single group for the whole set of moduli.
Clearly, for the better robustness, the way to group
the moduli or remainders plays a very important
role as one can see from the bounds above.
Note that the robustness bound $\tau<9/4$ in Corollary \ref{111}
for the single stage robust CRT
for the moduli set $\{180, 220, 486, 513\}$
is  only half of the robustness bound $\tau<18/4$
for the same moduli set $\{180, 220, 486, 513\}$ but
with the grouping and the two-stage
approach.
We next present our results for general cases.
First, we consider the case of two groups and two stages.


Let $\{M_{1,1},M_{1,2},\cdots, M_{1,L_1},M_{2,1},M_{2,2},\cdots,M_{2,L_2}\}$
be the whole set of moduli that may not be necessarily
all distinct. It is split to two groups with
Group 1 of $L_1$ moduli: $0<M_{1,1}< M_{1,2}<\cdots<M_{1,L_1}$;
and Gruop 2 of $L_2$ moduli:  $0<M_{2,1}<M_{2,2}<\cdots<M_{2,L_2}$.
These two groups do not have to be disjoint.
Let $N$ be a positive integer, and $r_{1,1}, r_{1,2}, \cdots, r_{1,L_1}, r_{2,1}, r_{2,2}, \cdots, r_{2,L_2}$ be the corresponding remainders of $N$, i.e.,
 \begin{equation}\label{34}
 N=n_{j,i}M_{j,i}+r_{j,i},
 \end{equation}
 where $0\leq r_{j,i}<M_{j,i}$ and $n_{j,i}$ is an unknown integer for $1\leq i\leq L_j$, $j=1\mbox{ or }2$.
 As we know, $N$ can be uniquely reconstructed from its $L_1+L_2$ remainders if and only if $0\leq N<\mbox{lcm}(\delta_1, \delta_2)$, where $\delta_1 \stackrel{\Delta}{=}\mbox{lcm}(M_{1,1}, M_{1,2}, \cdots, M_{1,L_1})$ and $\delta_2
\stackrel{\Delta}{=}\mbox{lcm}(M_{2,1}, M_{2,2}, \cdots, M_{2,L_2})$.
The congruence system (\ref{34}) can be converted into the following two-stage
congruences.


For $j=1, 2$, and Group $j$, we can write
\begin{equation}\label{1}
\left\{\begin{array}{ll}
N_j=K_{j,i}M_{j,i}+r_{j,i}\\
0\leq N_j<\delta_j\\
1\leq i \leq L_j.
\end{array}\right.
\end{equation}
Then, the above $N_1$ and $N_2$
can be combined to form a
 new system  of congruences:
\begin{equation}\label{3}
\left\{\begin{array}{ll}
N=l_1\delta_1+N_1\\
N=l_2\delta_2+N_2\\
0\leq N < \mbox{lcm}\left(\delta_1, \delta_2\right).
\end{array}\right.
\end{equation}
When $\delta_1=\delta_2$, the two congruence equations are degenerated
to a single equation and without loss of generality, we assume $\delta_1\neq \delta_2$ here
and such a similar degenerated case is not considered either  in what follows
in this paper.
Replacing $N_1$ and $N_2$ in (\ref{3}) by equations (\ref{1}), we have
\begin{equation}\label{4}
\left\{\begin{array}{ll}
N=l_1\delta_1+K_{1,i}M_{1,i}+r_{1,i},\,\,1\leq i\leq L_1\\
N=l_2\delta_2+K_{2,i}M_{2,i}+r_{2,i},\,\,1\leq i\leq L_2\\
0\leq N <\mbox{lcm}\left(\delta_1, \delta_2\right).
\end{array}\right.
\end{equation}
It is not hard to see that
$$
N=l_j\delta_j+K_{j,i}M_{j,i}+r_{j,i}=(l_j\frac{\delta_j}{M_{j,i}}+K_{j,i})M_{j,i}+r_{j,i},
$$
where
$$K_{j,i}<\frac{\delta_j}{M_{j,i}}
$$
for $1\leq i \leq L_j$, $j=1\mbox{ or }2$. Clearly,
$$
n_{j,i}=l_j\frac{\delta_j}{M_{j,i}}+K_{j,i},\,\,\,
1\leq i \leq L_j, \,\,\, j=1, 2.
$$

From the known  values of all the moduli $\{M_{j,i}\}$ and all the erroneous remainders $\{\tilde{r}_{j,i}\}$, if we can accurately determine $K_{j,i}$ and $l_j$, then we can accurately determine $n_{j,i}$.  Thus, we propose the following algorithm to robustly reconstruct $N$, called two-stage robust CRT,
 when the remainders are erroneous.

\begin{itemize}
  \item \textbf{Step 1:} Following the single stage robust CRT algorithm
of Steps 1-5 in Section \ref{sec2},
calculate $\hat{K}_{j,i}$ for $1\leq i\leq L_j$ in the system of congruence
equations (\ref{1}) from erroneous
remainders $\{\tilde{r}_i\}$
for each $j=1,2$.
  \item \textbf{Step 2:} After obtaining $\hat{K}_{j,i}$ for $1\leq i\leq L_j$, $j=1, 2$, calculate the average estimate $\hat{N}_j$ of $N_j$ for $j=1,2$
by equations (\ref{estimatevalue}) and (\ref{averagen}):
      \begin{equation}\label{39}
\hat{N}_j=\left[\frac{1}{L_j}\sum_{i=1}^{L_j}(\hat{K}_{j,i}M_{j,i}+\tilde{r}_{j,i})\right],
\end{equation}
where $\left[\cdot\right]$ stands for the rounding integer (\ref{round integer}).

  \item \textbf{Step 3:} Treating $\hat{N}_1$ and $\hat{N}_2$ as the new erroneous remainders in the system of congruence equations (\ref{3}) and following
the single stage robust CRT algorithm Steps 1-5 in Section
\ref{sec2} again, we  calculate $\hat{l}_1$ and $\hat{l}_2$.
  \item \textbf{Step 4:} Calculate $\hat{n}_{j,i}$ for $1 \leq i\leq L_j$
and $j=1,2$:
  \begin{equation}\label{determin folding}
  \hat{n}_{j,i}=\hat{l}_j\frac{\delta_j}{M_{j,i}}+\hat{K}_{j,i}.
  \end{equation}
  \item \textbf{Step 5:} Calculate the average estimate $\hat{N}$ of the
unknown integer $N$:
  \begin{equation}\label{determin aver}
  \hat{N}=\left[\frac{1}{L_1+L_2}\large\left(\sum_{i=1}^{L_1}\left(\hat{n}_{1,i}M_{1,i}+\tilde{r}_{1,i}\right)+\sum_{i=1}^{L_2}\left(\hat{n}_{2,i}M_{2,i}+\tilde{r}_{2,i}\right)\large\right)\right],
  \end{equation}
  where $\left[\cdot\right]$ stands for the rounding integer (\ref{round integer}).
\end{itemize}
Then, we have the following result. For $j=1,2$,
let $\tau_j$ denote the error level of the remainders $r_{j,i}$ in the j-th group, i.e.,
$$|\Delta r_{j,i}|=|\tilde{r}_{j,i}-r_{j,i}|\leq \tau_j,$$
 for $1\leq i \leq L_j$ and
 $$G_j=\max\limits_{1\leq i\leq L_j}\min\limits_{1\leq q\neq i\leq L_j}\frac{\mbox{gcd}\left(M_{j,i}, M_{j,q}\right)}{4}.
$$
Let
$$G=\frac{\mbox{gcd}\left(\delta_1, \delta_2\right)}{4}.
$$
In the above, when the j-th group has only one modulus $M_{j,1}$, then $G_j=\frac{M_{j,1}}{4}$ and the corresponding lcm, $\delta_j$, is just $M_{j,1}$.

\begin{theorem}\label{splitthem}
If
\begin{equation}\label{bbb0}
|\triangle r_{j,i}|\leq \tau_j< \min(G_j, G), \mbox{ for all } 1\leq i\leq L_j \mbox{ and } j=1,2,
\end{equation}
%
 then,  we can accurately determine the folding numbers $\hat{n}_{j,i}=n_{j,i}$ for $1\leq i \leq L_j$, $j=1, 2$, and
 the average estimate $\hat{N}$ of
the unknown integer $N$ in (\ref{determin aver}) satisfies
\begin{equation}
|\hat{N}-N|\leq \left[\frac{L_1\tau_1+L_2\tau_2}{L_1+L_2}\right],
\end{equation}
where $\left[\cdot\right]$ stands for the rounding integer (\ref{round integer}).
\end{theorem}
\begin{proof}
For $j=1,2$,
according to Corollary \ref{111}, when $|\Delta r_{j,i}|\leq \tau_j< G_j$, we can accurately determine $K_{j,i}$ in the systems of congruence equations (\ref{1}):
\begin{equation}
K_{j,i}=\hat{K}_{j,i},\mbox{ for }1\leq i \leq L_j.
\end{equation}
Furthermore, for  the average estimates  $\hat{N}_1$ and $\hat{N}_2$
in (\ref{39}) in Step 2 above, we have
$$|\Delta N_j|=|N_j-\hat{N}_j|\leq \tau_j,
$$
which keeps the same error level as the
remainders $r_{j,i}$ for $1\leq i\leq L_j$, $j=1, 2$.

In the second stage (\ref{3}), $\hat{N}_1$ and $\hat{N}_2$ become the erroneous remainders. To accurately determine $l_1$ and $l_2$,
according to Proposition \ref{wjwcorollary} or Corollary \ref{111},
the error levels should satisfy $\tau_1<G=\frac{\mbox{gcd}\left(\delta_1, \delta_2\right)}{4}$ and $\tau_2<G=\frac{\mbox{gcd}\left(\delta_1, \delta_2\right)}{4}$, then
\begin{equation}
l_{j}=\hat{l}_{j}, \mbox{ for } j=1, 2.
\end{equation}
Thus, combining with
the first stage,
we have the condition $\tau_1<\min(G_1, G)$ and $\tau_2<\min(G_2, G)$
so that $K_{j,i}=\hat{K}_{j,i}$ and $l_j=\hat{l}_j$ for $1\leq i\leq L_j$,
$j=1, 2$. Namely, we have $n_{j,i}=\hat{n}_{j,i}$ from (\ref{determin folding}).

After we accurately determine the folding numbers  $n_{j,i}=\hat{n}_{j,i}$ for $1\leq i\leq L_j$, $j=1, 2$, we can get the average estimate $\hat{N}$ in (\ref{determin aver}) of the unknown integer $N$, i.e.,
\begin{equation}
\begin{array}{lll}
\hat{N}&=&\displaystyle\left[\frac{1}{L_1+L_2}\large\left(\sum_{i=1}^{L_1}\left(\hat{n}_{1,i}M_{1,i}+\tilde{r}_{1,i}\right)+\sum_{i=1}^{L_2}\left(\hat{n}_{2,i}M_{2,i}+\tilde{r}_{2,i}\right)\large\right)\right]\\
&=&\displaystyle N+\left[\frac{1}{L_1+L_2}\large\left(\sum_{i=1}^{L_1}\left(\Delta r_{1,i}\right)+\sum_{i=1}^{L_2}\left(\Delta r_{2,i}\right)\large\right)\right].
\end{array}
\end{equation}
From  $|\Delta r_{j,i}|\leq \tau_j$ for $1\leq i\leq L_j$
for $j=1,2$,
we can easily obtain
\begin{equation}
|\hat{N}-N|\leq \left[\frac{L_1\tau_1+L_2\tau_2}{L_1+L_2}\right].
\end{equation}
Therefore the proof is completed.
\end{proof}

The above results for two groups of moduli can be easily generalized
to a general number of groups of moduli by using
Corollary \ref{111} twice for the two stages of the
congruence equations as follows.

Assume there are $s$ groups of moduli with $s>2$. For every $1\leq j\leq s$,
the elements in the j-th group are denoted as $0<M_{j,1}< M_{j,2}<\cdots<M_{j,L_j}$, let $\delta_j\stackrel{\Delta}{=}\mbox{lcm}\left(M_{j,1}, M_{j,2}, \cdots, M_{j,L_j}\right)$ and $\tau_j$ denote the error level of the remainders $r_{j,i}$,
  $1\leq i\leq L_j$,
from the moduli in the j-th  group,
and  when the j-th group has more than one element, define
\begin{equation}\label{gj}
G_j=\max\limits_{1\leq i\leq L_j}\min\limits_{1\leq q\neq i\leq L_j}\frac{\mbox{gcd}\left(M_{j,i}, M_{j,q}\right)}{4}.
\end{equation}
If the j-th group has only one element, $M_{j,1}$, define
$G_j=\frac{M_{j,1}}{4}$.
Let
\begin{equation}\label{gg1}
G=\max\limits_{1\leq i\leq s}\min\limits_{1\leq q\neq i\leq s}\frac{\mbox{gcd}\left(\delta_i, \delta_q\right)}{4}.
\end{equation}
Then, we have the following result.
\begin{theorem}\label{theeeeee}
If
\begin{equation}\label{bbbbbbb}
|\triangle r_{j,i}|\leq \tau_j< \min(G_j, G), \mbox{ for all } 1\leq i\leq L_j \mbox{ and } 1\leq j\leq s,
\end{equation}
then, we can accurately determine the folding numbers
 $\hat{n}_{j,i}=n_{j,i}$ for $1\leq i \leq L_j$ and  $1\leq j\leq s$, thus  we can robustly reconstruct $\hat{N}$ as an estimate of
 $N$ when $0\leq N<\mbox{lcm}\left(\delta_1, \delta_2, \cdots, \delta_s\right)$:
\begin{equation}\label{generalN}
  \hat{N}=\left[\frac{1}{\sum_{j=1}^{s}L_j} \sum_{j=1}^{s}
 \sum_{i=1}^{L_j} \left( \hat{n}_{j,i}M_{j,i}+\tilde{r}_{j,i}\right)
\right],
\end{equation}
and
\begin{equation}\label{generalE}
|\hat{N}-N|\leq \left[\frac{\sum_{j=1}^{s}L_j\tau_j}{\sum_{j=1}^{s}L_j}\right].
\end{equation}
\end{theorem}
\begin{proof}
The proof is similar to the proof of Theorem \ref{splitthem}.
In the first stage, via (\ref{bbbbbbb}) we can accurately determine the folding numbers $K_{j,i}$ for $1\leq i\leq L_j$, $1\leq j\leq s$, and obtain the robust estimates $\hat{N}_j$ of $N_j$ for the j-th group with the error bound $|\hat{N}_j-N_j|\leq \tau_j< \min(G_j, G)$, where $0\leq N_j<\delta_j$ for all $1\leq j\leq s$.

Then, in the second stage we take these estimates $\hat{N}_j$ as erroneous remainders and $\delta_j$ as moduli for $1\leq j\leq s$ to form a new congruence system. Applying the result of Corollary \ref{111} again, we can accurately determine the unknown folding numbers $l_j$ for $1\leq j\leq s$. By that, we can accurately determine $n_{j,i}$ for $1\leq i\leq L_j$ with $1\leq j\leq s$.

Lastly, once we get the accurate values of $n_{j,i}$, the average estimate $\hat{N}$ of $N$ can be found. The error bound of $N$ is proved
similar to  Theorem \ref{splitthem}. Hence, the theorem is proved.
\end{proof}

Similar to Corollary \ref{coroll2} for the single
stage robust CRT,
 in the second stage with moduli $\delta_j$ and erroneous remainders $\hat{N}_j$ for  $1\leq j\leq s$, we can also individually consider
the remainder error level for each remainder
$\hat{N}_j$
 with respect to modulus $\delta_j$
and have the following result.
\begin{corollary}\label{coroll3}
Assume that the reference modulus is $\delta_k$ for some $k$ with
$1\leq k\leq s$, which satisfies
$$G=\min\limits_{1\leq q\neq k\leq s} \frac{\mbox{gcd}\left(\delta_k, \delta_q\right)}{4}=
\max\limits_{1\leq i\leq s}\min\limits_{1\leq q\neq i\leq s} \frac{\mbox{gcd}\left(\delta_i, \delta_q\right)}{4}$$
 and its corresponding remainder error bound $\tau_k<G$. If
\begin{equation}
|\triangle r_{k,i}|\leq \tau_k<\min(G_k, G), \mbox{ for all } 1\leq i\leq L_k,
\end{equation}
\begin{equation}
|\triangle r_{j,i}|\leq \tau_j<\min(G_j, \frac{\mbox{gcd}(\delta_j, \delta_k)}{2}-\min(G_k, G)), \mbox{ for all } 1\leq i\leq L_j \mbox{ and } 1\leq j\neq k\leq s,
\end{equation}
then we can accurately determine the folding numbers $\hat{n}_{j,i}=n_{j,i}$ for $1\leq i\leq L_j$ and $1\leq j\leq s$.
\end{corollary}
\begin{example}
Given three groups of moduli: $\{26\cdot5, 26\cdot6\}$, $\{27\cdot7, 27\cdot13\}$, and $\{28\cdot15, 28\cdot11\}$. We can get $G_1=26/4$, $G_2=27/4$, $G_3=28/4$ and $G=21/4$. So, from Theorem \ref{theeeeee}, we obtain the robustness bounds: $\tau_1<\min(G_1, G)=21/4$, $\tau_2<\min(G_2, G)=21/4$ and $\tau_3<\min(G_3, G)=21/4$. If we use the result of Corollary \ref{coroll3}, we can get a better error bound for some groups  as follows: $\tau_1\leq26/4$, $\tau_2\leq21/4$ and $\tau_3<21/4$.
\end{example}

For a given set of $L$ moduli $M_i$,
$1\leq i\leq L$, there are
many different grouping methods of the remainders, and therefore,
many ways to robustly reconstruct the unknown integer
from its erroneous remainders.
Let us see an example.
\begin{example}
Consider the moduli set $\{M_1, M_2, M_3, M_4\}=\{48\cdot 4, 48\cdot 3, 56\cdot 3, 56\cdot 2\}$. First, consider all the moduli as a single group and
we obtain the robustness bound  $24/4$ according to Corollary
\ref{111}. Second, we group the four moduli into two groups $\{48\cdot4, 48\cdot3\}$ and $\{56\cdot3, 56\cdot2\}$ with $G_1=48/4$, $G_2=56/4$ and $G=\frac{\mbox{gcd}(\delta_1, \delta_2)}{4}=48/4$ in Theorem \ref{splitthem}. Accordingly, the robustness bound in this case is $\tau_1< 48/4$ and $\tau_2< 48/4$.
 Lastly, if we group the four moduli into another two groups $\{48\cdot4, 56\cdot2\}$ and $\{48\cdot3, 56\cdot3\}$ with $G_1=16/4$, $G_2=24/4$ and $G=336/4$, then, the robustness bound in this case is $\tau_1<16/4$ and $\tau_2<24/4$.
\end{example}

From this example, we can see that different grouping methods lead
to different robustness bounds. Compared with the robustness bound by using a single stage robust CRT for the whole set of moduli,
sometimes a  grouping can enlarge the robustness bound while
 sometimes a grouping may decrease the robustness bound. Thus,
 another question is  whether there exists a proper grouping method
 to ensure the robustness bound  larger than that in
Corollary \ref{111} using the single stage
robust CRT.
We next  present  a result that tells us when  there exists
a grouping method with a better robustness bound for remainders in some
groups
using a two-stage robust CRT than that using the single stage robust CRT.

\begin{corollary}\label{jielun}
For a given set of $L$ moduli $\{M_i$, $i=1, \cdots, L\}$,
the robustness bound can not be enlarged
for  remainders in any group by using a two-stage robust CRT with a
 grouping method of the moduli
 if and only if it is the case of \cite{wjwang},
i.e.,
the remaining factors $\Gamma_i$ of the moduli $M_i=M\Gamma_i$
divided by their
gcd $M$, $1\leq i\leq L$, are co-prime.
\end{corollary}

\begin{proof}
It is easy to prove the sufficiency as follows.
When the moduli $M_i$ satisfy the constraint in \cite{wjwang},
i.e., $\Gamma_i$ are co-prime,
its robustness bound using the
 single stage robust CRT with a single  group moduli is $M/4$.
On the other hand, from Theorem \ref{theeeeee}, each $G_j$ of any grouping
and $G$ are both $M/4$.
Hence, we cannot enlarge the robustness bound  in this  case.

We next prove the necessity.
Assume that the robustness bound
for remainders in any group
can not be enlarged by the
two-stage robust CRT with a grouping method of the moduli
over the  robustness bound of the single stage robust CRT of
the whole set of the  moduli.
Denote $\mbox{gcd}\left(M_i, M_q\right)=m_{iq}$.
Without loss of generality, we can assume
$$
\frac{m_{1L}}{4}=\frac{\mbox{gcd}\left(M_1, M_L\right)}{4}=\max\limits_{1\leq i\leq L}\min\limits_{1\leq q\neq i\leq L}\frac{\mbox{gcd}\left(M_i, M_q\right)}{4},$$
 and
$$m_{12}\geq m_{13}\geq \cdots \geq m_{1L}.
$$
 Thus,  according to Corollary \ref{111}, its robustness bound using
the single stage robust CRT with a single group
moduli is $m_{1L}/4$. We then have the following two cases.

\textbf{Case I:} There exists one $q$ with $2\leq q<L$ such that
$$m_{12}\geq \cdots \geq m_{1q}>m_{1(q+1)}\geq \cdots \geq m_{1L}.$$
If we group the moduli as
Group 1: $\{M_1, \cdots, M_q\}$;
and Group 2: $\{M_1, M_{q+1}, \cdots, M_L\}$.
With this grouping, we have that  $G_1>m_{1L}/4$ and $G_2\geq m_{1L}/4$,
 $G\geq M_1/4>m_{1L}/4$. Thus, we obtain $\tau_1<\min\{G_1, G\}$, $\tau_2<\min\{G_2, G\}$, where $\min\{G_1, G\}>m_{1L}/4=\tau$,
which contradicts with the assumption
 that we cannot enlarge the robustness bound for the remainders
in Group 1 using
a two-stage robust CRT. This proves that $m_{12}= m_{13}= \cdots = m_{1L}$.

\textbf{Case II:} Under the condition of $m_{12}= m_{13}= \cdots = m_{1L}=M$,
we know that any $m_{iq}=\mbox{gcd}(M_i, M_q)\geq M$, since $M$ is a factor of all the moduli $M_i$. Suppose that there exists one $m_{iq}>M$ with $q\neq i\neq 1$. We can group the moduli as Group 1: $\{M_i, M_q\}$,
and Group 2:  $\{M_i, \{M_{i_1}\}_{i_1\neq q \mbox{ or }i}\}$.
 Similar to Case I, $G_1>M/4$, $G_2\geq M/4$ and $G\geq M_i/4>M/4$. So,
 we can enlarge the robustness bound for the remainders in Group 1 by using the two-stage
robust CRT with this grouping. This also contradicts with the assumption.

From the above two cases we conclude that $m_{iq}=M$ for all $1\leq i\neq q\leq L$, i.e., it is the case of \cite{wjwang}.
\end{proof}

Now, we give an explicit example. Suppose that there are a set of  moduli with the form of $\{M_1K_1, M_2K_1, M_2K_2,\\ M_2K_3\}$, where $M_1, K_1, K_2, K_3$ are co-prime. According to Corollary \ref{111}, the robustness
bound using the single stage robust CRT is $\tau<\min\left(M_2, K_1 \mbox{gcd}\left(M_1, M_2\right)\right)/4$. If the moduli are grouped into  two groups as $\{M_1K_1, M_2K_1\}$ and $\{M_2K_1, M_2K_2, M_2K_3\}$. Then, according to Theorem \ref{splitthem},  its robustness bound is $\tau_1< K_1 \mbox{gcd}\left(M_1, M_2\right)/4$ and $\tau_2<M_2/4$, one of which
is greater than the robustness bound $\min(M_2, K_1 \mbox{gcd}(M_1,$\\$ M_2))/4$ when $M_2\neq K_1 \mbox{gcd}\left(M_1, M_2\right)$.

\begin{example}
Let $M_1=8$, $M_2=14$, $K_1=3$, $K_2=5$, and $K_3=7$. Then we can calculate $\tau_1<6/4$ and $\tau_2<14/4$ from the two-stage
robust CRT. One can see that $\tau_2<14/4$  is significantly
greater than $\tau<6/4$ using the single stage robust CRT.
\end{example}


From Corollary \ref{jielun}, one can see that as long as $\Gamma_i$ in moduli
$M_i$ are not all co-prime, using a two-stage robust CRT with some
 grouping method has a larger robustness bound for remainders in
some groups than  the single stage robust CRT does. In the same way,
we may treat $\{\delta_1, \delta_2, \cdots, \delta_s\}$
as a new set of moduli and group it again so that
the single stage robust CRT is applied three times with the following result.
We call it three-stage robust CRT.

Let us split  $\{\delta_1, \delta_2, \cdots, \delta_s\}$ in Theorem \ref{theeeeee} to $k$ groups. For every $1\leq t\leq k$, the elements in the t-th group are denoted as $0<\delta_{t,1}<\cdots<\delta_{t,y_t}$, let $\xi_t\stackrel{\Delta}{=}\mbox{lcm}(\delta_{t,1}, \cdots, \delta_{t,y_t})$ and define
$$\Upsilon_t=\max\limits_{1\leq i\leq y_t}\min\limits_{1\leq q\neq i\leq y_t}\frac{\mbox{gcd}(\delta_{t,i},\delta_{t,q})}{4}.$$
Let
$$\Upsilon=\max\limits_{1\leq i\leq k}\min\limits_{1\leq q\neq i\leq k}\frac{\mbox{gcd}(\xi_i,\xi_q)}{4}.$$
We then have the following result.
\begin{theorem}
If
\begin{equation}\label{three-stage condition}
|\triangle r_{j,i}|\leq \tau_j< \min(G_j, \min\limits_{t}\{\Upsilon_t: \delta_j\in \{\delta_{t,1}, \cdots, \delta_{t,y_t}\}\} , \Upsilon), \mbox{ for all } 1\leq i\leq L_j, 1\leq j\leq s \mbox{ and }1\leq t\leq k,
\end{equation}
then, we can accurately determine the folding numbers
 $\hat{n}_{j,i}=n_{j,i}$ for $1\leq i \leq L_j$ and  $1\leq j\leq s$, thus  we can robustly reconstruct $\hat{N}$ as an estimate of
 $N$ when $0\leq N<\mbox{lcm}\left(\delta_1, \delta_2, \cdots, \delta_s\right)$:
\begin{equation}\label{generalN1}
  \hat{N}=\left[\frac{1}{\sum_{j=1}^{s}L_j} \sum_{j=1}^{s}
 \sum_{i=1}^{L_j} \left( \hat{n}_{j,i}M_{j,i}+\tilde{r}_{j,i}\right)
\right],
\end{equation}
and
\begin{equation}\label{generalE1}
|\hat{N}-N|\leq \left[\frac{\sum_{j=1}^{s}L_j\tau_j}{\sum_{j=1}^{s}L_j}\right].
\end{equation}
\end{theorem}

\begin{proof}
The congruence system
$$N=n_{j,i}M_{j,i}+r_{j,i},$$
where $0\leq r_{j,i}<M_{j,i}$ for $1\leq i\leq L_j$, $1\leq j\leq s$ and $0\leq N<\mbox{lcm}(\delta_1, \delta_2, \cdots, \delta_s)$, can be converted into the following three-stage congruences.

For $1\leq j\leq s$, and Group $j$ in the first stage, we can write
\begin{equation}\label{11}
\left\{\begin{array}{ll}
N_j=K_{j,i}M_{j,i}+r_{j,i}\\
0\leq N_j<\delta_j\\
1\leq i \leq L_j.
\end{array}\right.
\end{equation}
In the second stage,
\begin{equation}\label{22}
\left\{\begin{array}{ll}
P_t=H_{t,i}\delta_{t,i}+N_{t,i}\\
0\leq P_t<\xi_t\\
1\leq i \leq y_t\\
1\leq t \leq k.
\end{array}\right.
\end{equation}
Then, in the third stage, we can write
\begin{equation}\label{33}
\left\{\begin{array}{ll}
N=l_t\xi_t+P_t\\
0\leq N<\mbox{lcm}(\xi_1, \cdots, \xi_k)\\
1\leq t \leq k.
\end{array}\right.
\end{equation}

As long as we can accurately determine all of $K_{j,i}$, $H_{t,i}$ and $l_t$ in each congruence system, we can then accurately determine $n_{j,i}$. According to conditions (\ref{three-stage condition}), we can accurately determine $K_{j,i}$ for $1\leq i \leq L_j$, $1\leq j\leq s$ and get the error bound
$$|\hat{N}_j-N_j|\leq \tau_j<\min(G_j, \min\limits_{t}\{\Upsilon_t: \delta_j\in \{\delta_{t,1}, \cdots, \delta_{t,y_t}\}\}, \Upsilon).$$

Next, in each  group of the second stage we take these estimates $\hat{N}_j$ as erroneous remainders and $\delta_j$ as moduli. Applying the result of Corollary \ref{111}, we can accurately determine $H_{t,i}$, and also get the robust estimate $\hat{P}_t$ satisfying
$$|\hat{P}_t-P_t|<\min(G_j,  \min\limits_{t}\{\Upsilon_t: \delta_j\in \{\delta_{t,1}, \cdots, \delta_{t,y_t}\}\},
\Upsilon).$$

 Similarly, treat the estimates $\hat{P}_t$ as the erroneous remainders and $\xi_t$ as moduli in the third stage. Since $|\hat{P}_t-P_t|<\Upsilon$, from Corollary \ref{111} again, we can accurately determine $l_t$.
Once we accurately determine these unknown folding numbers in each congruence
 system, we can accurately determine $n_{j,i}$
 and then obtain the robust estimate $\hat{N}$ of the unknown integer $N$.
As for the error bound of the estimate $\hat{N}$, the proof is the same to that of Theorem \ref{splitthem}. Therefore, we complete the proof.
\end{proof}

\begin{example}
Consider a given set of moduli $\{96\cdot2, 96\cdot3, 72\cdot3, 72\cdot5, 64\cdot5, 64\cdot7\}$. Treating them as one group and using the single robust CRT, we get its error bound $\tau$ for the remainders satisfying $\tau<32/4$. If we split the moduli to three groups: $\{96\cdot2, 96\cdot3\}$, $\{72\cdot3, 72\cdot5\}$ and $\{64\cdot5, 64\cdot7\}$, we  get $G_1=96/4$, $G_2=72/4$, $G_3=64/4$, $\delta_1=96\cdot2\cdot3$, $\delta_2=72\cdot3\cdot5$, $\delta_3=64\cdot5\cdot7$ and $G=64/4$. By using the two-stage robust CRT, we can get the error bounds $\tau_j$ for the remainders in Group $j$ for $j=1, 2, 3$ satisfying $\tau_1<\min(G_1, G)=64/4$, $\tau_2<\min(G_2, G)=64/4$ and $\tau_3<\min(G_3, G)=64/4$, all of which are larger than the bound $32/4$ in the single robust CRT. If we use the three-stage robust CRT and split $\{\delta_1, \delta_2, \delta_3\}$ to two  groups again: $\{\delta_1, \delta_2\}$ and $\{\delta_3\}$. We can get $\Upsilon_1=72/4$, $\Upsilon=320/4$. So, in this three-stage robust CRT,  the error bounds satisfy $\tau_1<\min(G_1, \Upsilon_1, \Upsilon)=72/4$, $\tau_2<\min(G_2, \Upsilon_1, \Upsilon)=72/4$ and $\tau_3<\min(G_3, \Upsilon)=64/4$. Compared with the two-stage robust CRT, we increase the robustness bounds in Group $1$ and Group $2$ from $64/4$ to $72/4$ by using the three-stage robust CRT.
\end{example}

The above three-stage robust
CRT can be easily generalized to a multi-stage robust CRT
with more than three stages.
Although we can use a multi-stage robust CRT with some grouping methods to obtain a larger robustness bound for remainders in some groups, there are some challenges about how to choose
moduli in a group and how many groups and stages we should split in order to find a better robustness bound such that we can enlarge all the robustness bounds in every group.

Let us first look at the simplest case when there are only three moduli $\{M_1, M_2, M_3\}$. Without loss of generality, we can assume that $\mbox{gcd}\left(M_1, M_2\right)\geq \mbox{gcd}\left(M_1, M_3\right)\geq \mbox{gcd}\left(M_2, M_3\right)$. Regarding the three moduli as one group and by Corollary \ref{111}, the robustness bound is $\mbox{gcd}\left(M_1, M_3\right)/4$. Since $\mbox{gcd}\left(M_3, M_2\right)\leq \mbox{gcd}\left(M_3, M_1\right)$, if we want to obtain a robustness bound
strictly larger than $\mbox{gcd}\left(M_1, M_3\right)/4$, the modulus $M_3$ must independently form an individual group by itself, and in the meantime it does not allow other groups to include $M_3$. Thus, there is only one possible grouping method
 as $\{M_3\}$ and $\{M_1, M_2\}$. The robustness bound
therein is  $\tau_1<\mbox{gcd}\left(M_3, \mbox{lcm}\left(M_1, M_2\right)\right)/4$ and $\tau_2<\min\left(\mbox{gcd}\left(M_1, M_2\right), \mbox{gcd}\left(M_3, \mbox{lcm}\left(M_1, M_2\right)\right)\right)/4$, which may be both larger than $\mbox{gcd}\left(M_1, M_3\right)/4$. Otherwise, we have to group them as
 $\{M_1, M_3\}$ and $\{M_1, M_2\}$ and in this way we may only enlarge one
group's (not all group's) robustness bound as what is used in the proof of Corollary \ref{jielun}.

\begin{example}
When $M_1=560, M_2=480$ and $M_3=210$, we can see that $\mbox{gcd}\left(M_1, M_2\right)=\mbox{gcd}\left(560, 480\right)=80$,  $\mbox{gcd}\left(M_1, M_3\right)=\mbox{gcd}\left(560, 210\right)=70$ and  $\mbox{gcd}\left(M_2, M_3\right)=\mbox{gcd}\left(480, 210\right)=30$. Regarding these three moduli as a single
group, the robustness bound of the single stage robust CRT
 is $\mbox{gcd}\left(M_1, M_3\right)/4=\mbox{gcd}\left(560, 210\right)/4=70/4$. In order to find a larger robustness bound, we just only consider the robustness bound of the case of two groups: $\{M_3\}$ and $\{M_1, M_2\}$. We can get $\tau_1<210/4$ and $\tau_2<80/4$, which are all larger than $70/4$.
\end{example}

The above special case is about only three moduli's grouping. When the number of given moduli is larger, it  becomes more complicated. In the next section, we analyze some special cases.


\section{An Algorithm for Grouping Moduli in Two-Stage Robust CRT}\label{sec4}

From the above study, one may see that for a given set of moduli,
although its determinable range for an integer from its remainders
is fixed, i.e., the lcm of all the moduli,
the robustness bounds for an erroneous remainder
and a reconstructed integer
depend on a reconstruction algorithm from erroneous remainders,
which depends on the grouping of the moduli in a multi-stage
robust CRT. For a general set of moduli, it is not obvious on how to group them in a multi-stage (or even two-stage) robust CRT, in particular when
the number of moduli is not small.
In this section,
 based on Theorem \ref{theeeeee} for the two-stage robust CRT,
 we propose an algorithm for grouping a general set of moduli to possibly
obtain
a larger robustness bound for remainders in every group than that
in Corollary \ref{111} for the single stage robust CRT.

For a given set of moduli ${\cal M}=\{M_1, M_2, \cdots, M_L\}$, $L\geq3$, we
first assume that the set of moduli does not include any pair of $M_{i_1}$ and $M_{i_2}$ satisfying $M_{i_1}=nM_{i_2}$, because Corollary \ref{cor2} has told us that such a redundant modulus $M_{i_2}$ does not help to increase
the  determinable range of $N$, $0\leq N<\mbox{lcm}(M_1, M_2, \cdots, M_L)$ nor the robustness bound
in a single stage robust CRT. From condition (\ref{bbbbbbb})
we need to assure that all $G_j$ in (\ref{gj}) and $G$ in (\ref{gg1}) after a grouping strictly greater than
$\Theta\stackrel{\Delta}{=}\max\limits_{1\leq i\leq L}\min\limits_{1\leq j\neq i\leq L} \frac{\mbox{gcd}\left(M_i, M_j\right)}{4}$ in Corollary \ref{111} for  the single stage robust CRT. Then, we have an algorithm as follows.



\begin{enumerate}
  \item For each $M_i$, $1\leq i\leq L$, find all $M_j$, $1\leq j\neq i\leq L$, satisfying $\frac{\mbox{gcd}(M_{j}, M_i)}{4}>\Theta$. With $M_i$, form the corresponding set ${\cal M}_i$:
 $${\cal M}_i=\{M_i, M_{j}: \frac{\mbox{gcd}(M_{j}, M_i)}{4}>\Theta\}.$$
Thus, with each set ${\cal M}_i$, we have
$$G_{i}\geq \min\limits_{M_{j}\in {\cal M}_i}\frac{\mbox{gcd}(M_{j}, M_i)}{4}>\Theta.$$
If  modulus $M_i$ satisfies  $\frac{\mbox{gcd}(M_{j}, M_i)}{4}\leq\Theta$ for all $M_j$, $1\leq j\neq i\leq L$, then we let ${\cal M}_i=\{M_i\}$.
  \item Among all of the $L$ sets ${\cal M}_i$ for $1\leq i\leq L$, there may be one or more pairs, ${\cal M}_{i_1}$ and ${\cal M}_{i_2}$, satisfying ${\cal M}_{i_1}\subseteq {\cal M}_{i_2}$. In this case,
   we can delete the smaller
  set ${\cal M}_{i_1}$ and only keep the larger set ${\cal M}_{i_2}$.
  \item After Step 2), from the remaining sets of $\{{\cal M}_{i}\}$, we find all such combinations of $\{{\cal M}_{i_1}, {\cal M}_{i_2}, \cdots, {\cal M}_{i_l}\}$ that $\bigcup\limits_{j=1}\limits^{l}{\cal M}_{i_j}$ exactly includes all moduli ${\cal M}$. In other words,
if anyone ${\cal M}_{i_s}$ for $1\leq s\leq l$
is deleted from a combination $\{{\cal M}_{i_1}, {\cal M}_{i_2}, \cdots, {\cal M}_{i_l}\}$, then $\bigcup\limits_{j=1 \& j\neq s}\limits^{l}{\cal M}_{i_j}$ is a proper subset of ${\cal M}$, i.e.,
      \begin{equation}\label{combinat}
{\cal M}\neq     \bigcup\limits_{j=1 \& j\neq s}\limits^{l}{\cal M}_{i_j}\subset {\cal M} \subseteq \bigcup\limits_{j=1}\limits^{l}{\cal M}_{i_j}.
      \end{equation}
  \item As for every combination in the above, treat each ${\cal M}_{i_j}$ as a small group and calculate its lcm as $\delta_{i_j}$, $1\leq j\leq l$. Then, check whether
  \begin{equation}\label{gggg}
  G=\max\limits_{1\leq p\leq l}\min\limits_{1\leq q\neq p\leq l}\frac{\mbox{gcd}\left(\delta_{i_{p}}, \delta_{i_{q}}\right)}{4}>\Theta.
   \end{equation}
   If there is one combination
 $\{{\cal M}_{i_j}\}$ as above
to make inequality (\ref{gggg}) hold, then every $\min(G_{i_j}, G)$, $1\leq j\leq l$, is strictly greater than $\Theta$. According to (\ref{bbbbbbb}) in Theorem \ref{theeeeee}, one can see that this combination is just a grouping as desired and it enlarges a robustness bound for remainders in every group by using the two-stage robust CRT. Otherwise, if for every possible
 combination in Step 3), inequality (\ref{gggg}) does not hold, then it is said that we fail to use this algorithm to enlarge a robustness bound for remainders in every group by using the two-stage robust CRT.
  \end{enumerate}

Let us first consider the above grouping algorithm for the case of \cite{wjwang}, i.e.,
the remaining factors $\Gamma_i$ of the moduli $M_i=M\Gamma_i$
divided by their gcd $M$, $1\leq i\leq L$, are co-prime. First, we find all ${\cal M}_i=\{M_i\}$, $1\leq i\leq L$. Next, there is only one combination
$\{{\cal M}_{1}, {\cal M}_{2}, \cdots, {\cal M}_{L}\}$ satisfying (\ref{combinat}), and we treat each ${\cal M}_i=\{M_i\}$ as one group, then calculate $G=M/4$ in (\ref{gggg}), which equals to $\Theta=M/4$. In conclusion, we fail to find a grouping  to enlarge a robustness bound for remainders in every group by using the two-stage robust CRT,
which can be also confirmed from the earlier result in Corollary \ref{jielun}.
Next, we give a positive example.

\begin{example}
Consider a set of moduli $\{210M, 143M, 77M, 128M, 81M, 125M, 169M\}$, where $M$ is an integer. As one group, using the
single stage robust CRT, its robustness bound is $\Theta=M/4$.
According to the above grouping algorithm, find $7$ sets: ${\cal M}_1=\{210M, 77M, 128M, 81M, 125M\}$, ${\cal M}_2=\{143M, 77M, 169M\}$, ${\cal M}_3=\{77M, 210M, 143M\}$, ${\cal M}_4=\{128M, 210M\}$, ${\cal M}_5=\{81M, 210M\}$, ${\cal M}_6=\{125M, 210M\}$ and ${\cal M}_7=\{169M, 143M\}$. Among them, there are only four combinations satisfying (\ref{combinat}) as follows: $\{{\cal M}_1, {\cal M}_2\}$, $\{{\cal M}_1, {\cal M}_7\}$, $\{{\cal M}_2, {\cal M}_4, {\cal M}_5, {\cal M}_6\}$ and $\{{\cal M}_3, {\cal M}_4, {\cal M}_5, {\cal M}_6, {\cal M}_7\}$. Then, check whether one of the above four combinations satisfies inequality (\ref{gggg}). Fortunately, for the first combination $\{{\cal M}_1, {\cal M}_2\}$, inequality (\ref{gggg}) holds. We can calculate $G_1=2M/4$, $G_2=11M/4$ and $G=7M/4$, all of which are strictly greater than $M/4$. Thus, we have obtained a grouping method of the moduli to enlarge a robustness bound for remainders in every group by using the two-stage robust CRT.
\end{example}

\begin{remark}
As one can see in the proof of Corollary \ref{jielun}
and in the above algorithm and examples, a modulus $M_i$ may be repeatedly used in more than one groups in the two-stage robust CRT. Its aim is to make $G$ and $G_j$ after grouping greater than or equal to the robustness bound by using
the  single stage robust CRT for the whole set of moduli. Recall the case of grouping a set of three moduli $\{M_1, M_2, M_3\}$. Assume $\mbox{gcd}\left(M_1, M_2\right)> \mbox{gcd}\left(M_1, M_3\right)> \mbox{gcd}\left(M_2, M_3\right)$. From Corollary \ref{111}, the robustness bound for using the single robust CRT is $\frac{\mbox{gcd}\left(M_1, M_3\right)}{4}$. According to the above grouping moduli algorithm in two-stage robust CRT, they are split to two groups: $\{M_1, M_2\}$ and $\{M_3\}$. One can see that $G_1=\frac{\mbox{gcd}(M_1, M_2)}{4}$, $G_2=\frac{M_3}{4}$, $\delta_1=\mbox{lcm}(M_1, M_2)$, $\delta_2=M_3$
and $G=\frac{\mbox{gcd}(\delta_1, \delta_2)}{4}$. In this grouping method, the robustness bound for remainders in group $\{M_1, M_2\}$ is $\min(G_1, G)$ and the robustness bound for remainders in group $\{M_3\}$ is $\min(G_2, G)$. As $G_1$ and $G_2$ are greater than $\frac{\mbox{gcd}\left(M_1, M_3\right)}{4}$, a robustness bound for remainders in each group depends on the value of $G$.
When $G=\frac{\mbox{gcd}(\delta_1, \delta_2)}{4}$ is less than $\frac{\mbox{gcd}\left(M_1, M_3\right)}{4}$, a robustness bound for remainders in each group is worse than that for the single robust CRT. Thus, we should repeat modulus $M_1$ in group $\{M_3\}$, and the two groups become $\{M_1, M_2\}$ and $\{M_1, M_3\}$. In this way, we enlarge a robustness bound for group $\{M_1, M_2\}$ and keep the robustness bound for group $\{M_1, M_3\}$ non-changed. On the other hand, when $G=\frac{\mbox{gcd}(\delta_1, \delta_2)}{4}$ is larger than $\frac{\mbox{gcd}\left(M_1, M_3\right)}{4}$, we do not need to repeat modulus $M_1$, since the robustness bound for group $\{M_1, M_2\}$ and the robustness bound for group $\{M_3\}$ are both greater than $\frac{\mbox{gcd}\left(M_1, M_3\right)}{4}$.
This example tells us that, to enlarge the robustness bound,
whether a modulus $M_i$
is repeatedly used or not in multiple groups depends on the grouping
method and the set of moduli. Repeating a modulus, sometimes,  may help
to enlarge the robustness bound but sometimes may not.
\end{remark}

\section{Simulations}\label{sec5}

In this section, we present some simple
simulation results to evaluate the proposed single stage robust CRT algorithm and the two-stage robust CRT algorithm for integers with a general set of moduli. Let us first consider the case when $M_1=9\cdot15$, $M_2=9\cdot20$ and $M_3=9\cdot18$. These three moduli do not satisfy the condition
that $\Gamma_i$, $i=1,2,3$, are co-prime and thus the robust CRT obtained
in \cite{xwli2, wjwang} can not be applied directly.
However, we can use  our proposed single stage robust CRT.
According to Corollary \ref{111}, the maximal range of
the determinable $N$ is $1620$ and the maximal remainder error level $\tau$
for the robustness is upper bounded by $\tau<\frac{27}{4}$ from (\ref{tttttt}).
In this simulation, the unknown integer $N$ is uniformly distributed in the interval $[0, 1620)$. We consider the maximal remainder error levels $\tau=0, 1, 2, 3, 4, 5, 6$, and the errors are also uniformly distributed
on  $[0,\tau]$
in the remainders. $2000000$ trials for each of them are implemented.
The mean error $E(|\hat{N}-N|)$ between the estimated $\hat{N}$ in (\ref{averagen}) and the true $N$ is plotted by the solid line marked with $\square$, and
the theoretical estimation error upper bound in (\ref{bound of fre}) is plotted by the solid line marked with $\triangle$ in Fig. \ref{figone}.
Obviously, one can see that for a general set of moduli the reconstruction errors of $N$ from the erroneous remainders are small compared to the range of $N$.

Next, we compare the robustness between the single stage and the two-stage
robust CRT algorithms for the above same set of moduli.
In this case, the conditions of the maximal remainder error levels for the
single stage and the two-stage robust CRT algorithms of two groups
$\{M_1, M_2\}$ and $\{M_3\}$
are
$\frac{27}{4}$ and $\frac{45}{4}$, i.e., $\tau\leq 6$ and  $11$,
respectively. Let us
consider the maximal remainder error levels $\tau$ from $0$ to $25$, and $2000000$ trials for each of them.
The unknown integer $N$ is taken as before.
Fig. \ref{figtwo} shows the curves of
the error bounds and the mean estimation errors $E(|\hat{N}-N|)$
for both the single stage and the two-stage robust CRT algorithms.
Note that from our single stage robust CRT theory, the valid
error bound for $\tau$ is only upto $6$, which can be seen from the
simulation results that
the  mean estimation  error $E(|\hat{N}-N|)$ starts to deviate the previous line trend at  $\tau=7$, then increases significantly and
breaks the  linear error bound when $\tau$ is further greater, i.e., robust reconstruction may not hold. On the other hand, with the two-stage robust CRT algorithm, one can see that the curve of the mean estimation
 error $E(|\hat{N}-N|)$ is always below the curve of the
error bound, i.e., we can robustly reconstruct $N$, $0\leq N<1620$, even when
the maximal error level is  $11$ that is the upper bound for $\tau$ obtained
in this paper for the two-stage robust CRT algorithm.
These simulation results confirm the theory obtained in this paper.

\section{Conclusion}\label{sec6}

In this paper, we considered the robust reconstruction problem from
erroneous remainders, namely robust CRT problem, for a general set of
moduli that may not satisfy the condition needed in the previous
robust CRT studies in \cite{xwli2, wjwang}.
We obtained a necessary and sufficiency condition for the robust CRT
when all the erroneous remainders are used together, called single stage
robust CRT.
 Interestingly,  our proposed single stage robust CRT may
have  better robustness than that
of the robust CRT obtained in \cite{xwli2, wjwang} even when it
could be applied.
To further improve the robustness, we then proposed a multi-stage
robust CRT, where the moduli are grouped into several groups.
As an example, for the two-stage robust CRT,
our proposed single stage robust CRT is first applied to each group
and then applied across the groups second time.
Also, an algorithm on how to group a given set of moduli was proposed.
We finally presented some  simulations to verify  our proposed
theory.

\begin{figure}
  \includegraphics[width=\textwidth]{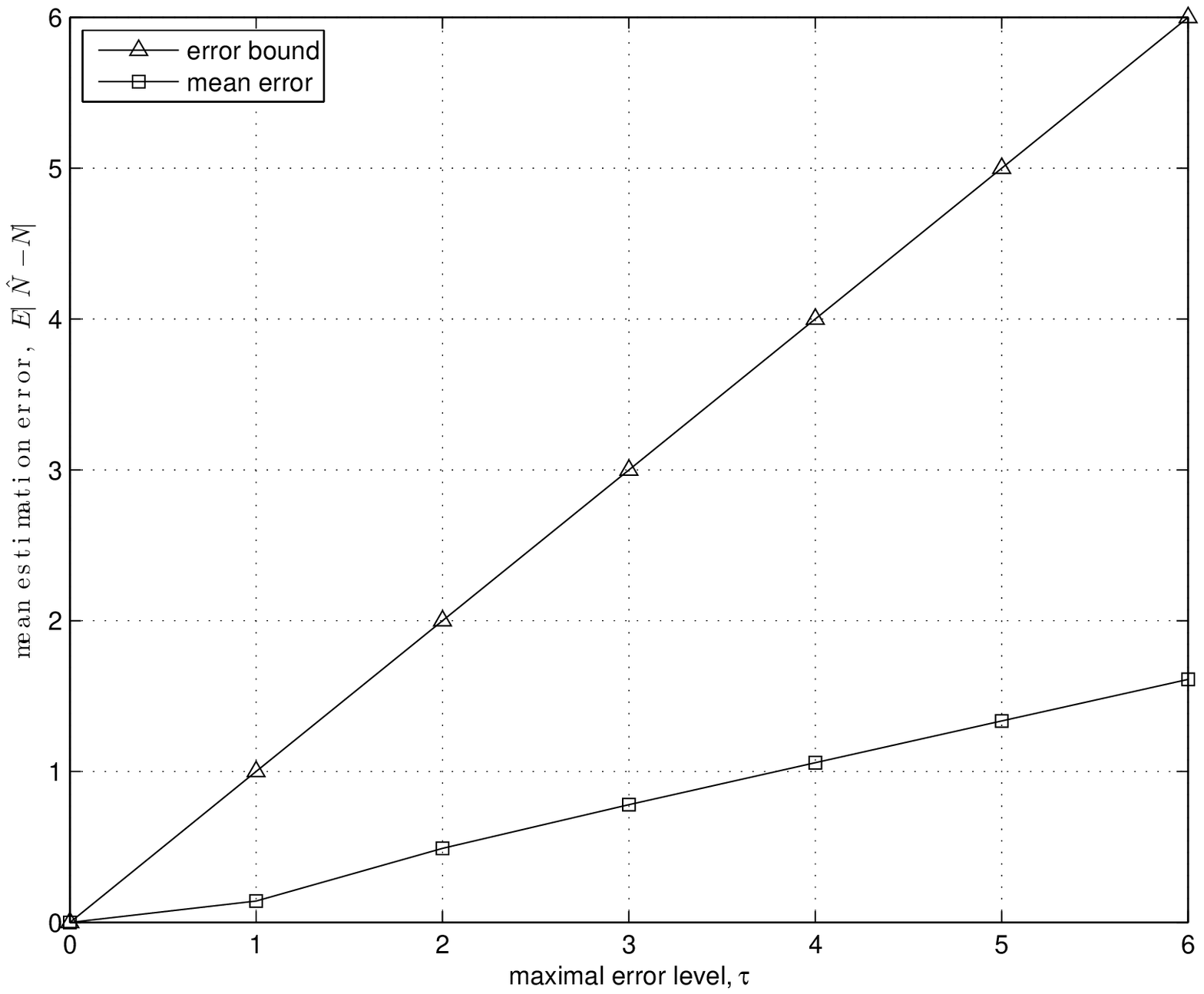}\\
  \caption{Mean estimation error and theoretical error
bound using the single stage robust CRT}\label{figone}
\end{figure}

\begin{figure}
  \includegraphics[width=\textwidth]{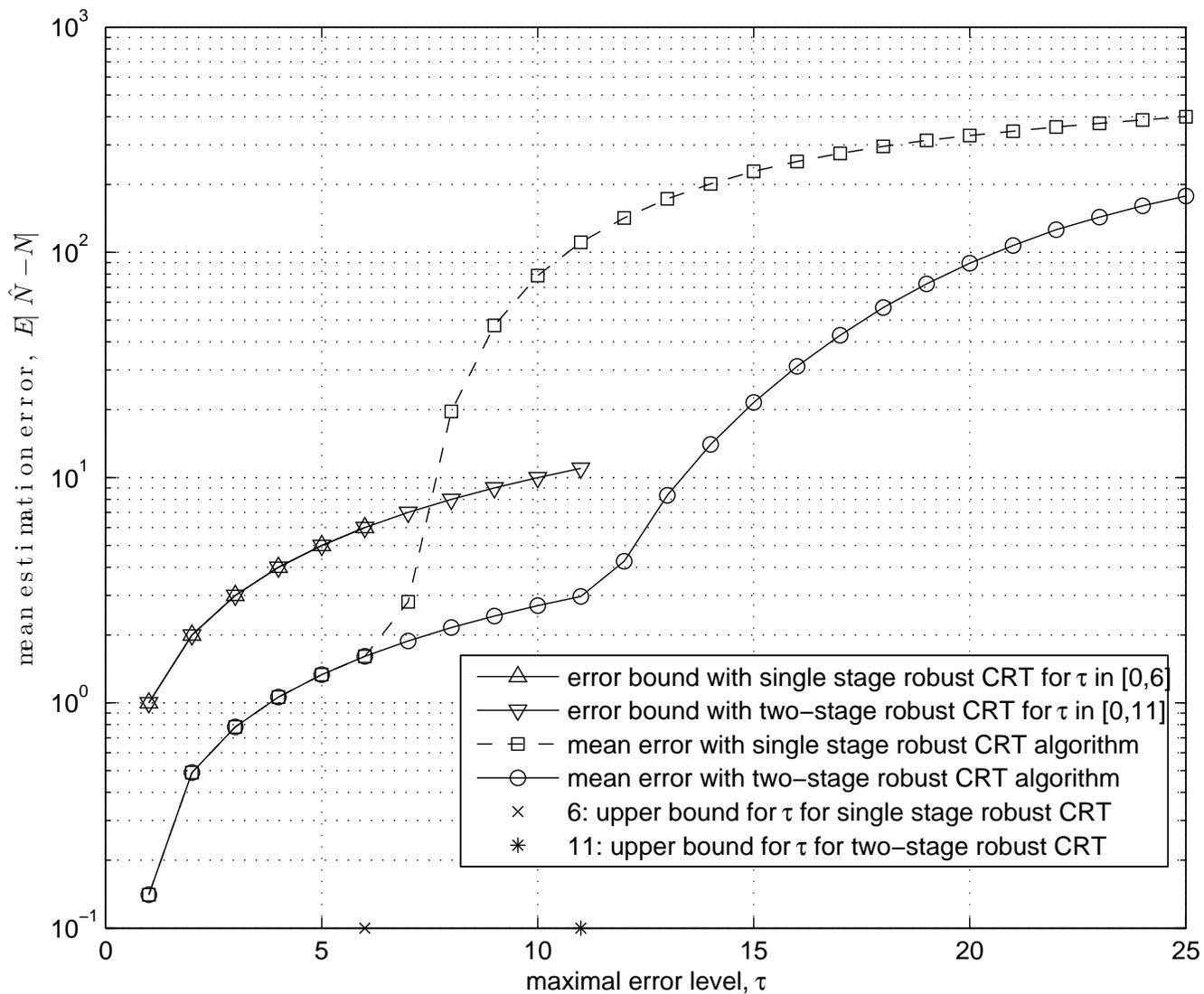}\\
  \caption{Mean estimation error and theoretical error bound
 comparison using the single stage robust CRT and the two-stage robust CRT}\label{figtwo}
\end{figure}

\end{document}